\documentclass{svproc}
\usepackage{amsmath}
\usepackage{amssymb}
 \usepackage{paralist}
 \usepackage{graphics} 
 \usepackage{epsfig} 
\usepackage{graphicx}  \usepackage{epstopdf}
 \usepackage[colorlinks=true]{hyperref}
\usepackage{graphicx}
\usepackage{epsf}
\usepackage{float}
\usepackage[dvipsnames]{xcolor}
\usepackage{booktabs}
\usepackage{color}
\usepackage{subfigure}
\usepackage{multicol}
\usepackage{enumerate}
\newcommand{\br}{\vspace{10pt} \newline}

\newcommand{\bs}{\backslash}
\newcommand{\bx}{{\boldsymbol{x}}}
\newcommand{\by}{{\boldsymbol{y}}}
\newcommand{\bz}{{\boldsymbol{z}}}
\newcommand{\bu}{{\boldsymbol{u}}}
\newcommand{\bv}{{\boldsymbol{v}}}

\newcommand{\ip}[2]{\langle #1,#2\rangle}

\newcommand{\rE}{\mathrm{E}}
\newcommand{\rL}{\mathrm{L}}

\newcommand{\tr}{\mathrm{Tr}}

\newcommand{\ddt}[1]{\frac{\mathrm{d}#1}{\mathrm{d}t}}

\newcommand{\be}{\begin{equation}}
\newcommand{\ee}{\end{equation}}
\newcommand{\bea}{\begin{eqnarray}}
\newcommand{\eea}{\end{eqnarray}}
\newcommand{\beas}{\begin{eqnarray*}}
\newcommand{\eeas}{\end{eqnarray*}}
\newcommand{\bl}{\begin{law}}
\newcommand{\el}{\end{law}}
\newcommand{\bthm}{\begin{thm}}
\newcommand{\ethm}{\end{thm}}

\newcommand{\ud}{\,\mathrm{d}}

\newcommand{\TT}{\mathrm{T}}

\newcommand{\bd}[1]{\boldsymbol{#1}}

\newcommand{\abs}[1]{\lvert#1\rvert}

\newcommand{\I}{\mathrm{i}}

\newcommand{\bP}{\boldsymbol{P}}
\newcommand{\bQ}{\boldsymbol{Q}}

\newcommand{\bp}{\boldsymbol{p}}
\newcommand{\bq}{\boldsymbol{q}}

\newcommand{\ba}{\mathbf a}

\newcommand{\mC}{\mathbb C}

\newcommand{\bF}{\mathbf F}

\newcommand{\mN}{\mathbb N}

\newcommand{\mR}{\mathbb R}

\newcommand{\cF}{\mathcal{F}}

\newcommand{\cI}{\mathcal{I}}

\newcommand{\cL}{\mathcal{L}}

\newcommand{\cO}{\mathcal{O}}

\newcommand{\cS}{\mathcal{S}}

\newcommand{\cZ}{\mathcal{Z}}

\newcommand{\barint}{\kern4pt \raise3.4pt\hbox{\vrule height.6pt
    width7pt} \kern-11pt \int}

\newcommand{\tp}{\mathrm{p}}
\renewcommand{\ts}{\mathrm{s}}
\newcommand{\tsh}{\mathrm{sh}}
\newcommand{\tsv}{\mathrm{sv}}
\newcommand{\FGA}{\mathrm{F}}

\newcommand{\Id}{\mathrm{Id}}

\DeclareMathOperator{\Curl}{\nabla\times}
\DeclareMathOperator{\Div}{\nabla\cdot}
\newcommand{\ds}{\displaystyle}

\numberwithin{equation}{section}
\def\main{\par\noindent{\bf Main Theorem.\ } \ignorespaces}
\def\endmain{}

\usepackage{url}

\begin{document}
\mainmatter    

\title{Convergence of the Frozen Gaussian Approximation for High Frequency Elastic Waves}

\titlerunning{Convergence of the FGA for High Frequency Elastic Waves}

\author{James C. Hateley \inst{1}\and Xu Yang\inst{2}}

\institute{Vanderbilt University, Nashville TN 37240, USA,\\
\email{james.c.hateley@vanderbilt.edu},
\and
UC Santa Barbara, Santa Barbara CA 93117, USA \\
\email{xuyang@math.ucsb.edu}
}

\maketitle

\begin{abstract}
The frozen Gaussian approximation (FGA) is an effective tool for modeling high frequency wave propagation.  In previous works, the convergence of the FGA has established for strict hyperbolic systems. In this work, we derive the frozen Gaussian approximation for the elastic wave equation, which can be cast as a hyperbolic system with repeated eigenvalues.  In the derivation, the strong form for the evolution equation is introduced.  A diabatic coupling is observed for the amplitude of the evolution equations between the SH, SV waves. Using previous results with new energy estimates we establish the convergence for the first order FGA for the elastic wave equation.
\end{abstract}
\section{Introduction}
This paper is to analytical show that the first ordered frozen Gaussian approximation for the elastic wave equation derived and reported in~\cite{hate2018fga} is indeed the same order for strict hyperbolic systems, as proven in~\cite{LuYa:MMS} and verified numerical in~\cite{LuYa:11,chai2017frozen,chai2018tomo,hate2018fga}.  The proof will use all the same machinery from~\cite{LuYa:MMS}, with a new derivation and new error estimates for the evolution equation. The theorem that will be proven is: \\
\main {\itshape
Let $\{\bu_0^\epsilon\}$ be a family of asymptotically high frequency initial conditions for displacement, let $\bu:[0,T]\times \mR^3$ satisfy the elastic wave equation~\eqref{cp:eq:1}, with $\|\bu_0^\epsilon\|_{\rL^2}s \in H_0^1(\mR^{3d})$ and uniformly bounded, i.e. $\|\bu_0^\epsilon\|_{H_0^1} < M$ and let $\bu_\FGA$ be the FGA approximation in~\eqref{eq:fga_sol}, with the proper assumptions
\begin{equation*}
\sup_{t\in[0,T]}\|\bu(t,\cdot) - \bu_{\FGA,0}(t,\cdot)\|_{\rE} \lesssim \epsilon M
\end{equation*}
Where the norm $\|\cdot\|_{\rE}$ is a scaled semi-norm defined in~\eqref{eq:enorm}. 
}
\endmain \\

As a corollary to this, we will show asymptotic equivalence to the formulation of the FGA in~\cite{LuYa:MMS}. \\

\noindent{\bf Related works.} The FGA is originally motivated by quantum chemistry and the herman-Kluk propagater~\cite{Ka:06,SwRo:09,Ro:10}. The FGA high frequency waves has been used for forward modeling in seismic tomography for the acoustic (P-waves)~\cite{chai2017frozen,chai2018tomo} and elastic waves (P,S-waves)~\cite{hate2018fga}. The FGA for the elastic wave equation has also been used train neural networks for seismic interface and pocket detection~\cite{hate2018nn}. \\

\noindent{\bf Organization of the paper.} In section~\ref{sec:2}, the necessary notation is introduced with preliminaries for phase plane analysis. In section~\ref{sec:3} the equation and it's eigenvalue decomposition is are introduced. Section~\ref{sec:4} describes the formulation and derivation of the frozen Gaussian approximation for the elastic wave equation. Finally, in section~\ref{sec:5} the error estimates and proofs are given.
\section{Preliminaries}\label{sec:2}

\subsection{Notation.}
We denote $\bx,\by \in \mR^d$ as spatial variables, $(\bq,\bp)\in\mR^{2d}$ for the position and momentum variables, respectively, in the phase space. \\

For $\delta > 0$, define the closed set $K_\delta \subset \mR^{2d}$
\begin{equation}
K_\delta = \Big\{ (\bq,\bp)\in \mR^{2d} : |\bq|\leq 1/\delta,\ \ \delta\leq |\bp| \leq 1/\delta\Big\}
\end{equation}
For all practical purpose, this set $K_\delta$ is bounding the position and the magnitude for the direction of propagation of the wave packets.  The upper bound conditions on $|\bq|$ and $|\bp|$ are quite reasonable as any computational domain will be a finite domain.  Also, $\bp$ bounded away from zero is reasonable as if $\bp = 0$ the wave packet does not propagate and the Hamiltonian system is degenerate i.e., $H = 0$.  \\

We use the notation $\cO(\epsilon^\infty)$: $A^\epsilon = \cO(\epsilon^\infty)$ meaning for any $k\in \mN$
\begin{equation}
\lim_{\epsilon\to 0}\epsilon^{-k}|A^{\epsilon}| = 0.
\end{equation}
Notation $C$ will be used as a general constant, that can vary from line to line.  The value will not be important, but rather that it is finite.  We will use subscripts to denote constant dependence, e.g. $C_T$, is a constant that depends on the parameter $T$. We will use the usual notation of $\cS$, $C^\infty$ and $C^\infty_c$ for the Schwartz class, smooth and compacted supported smooth functions respectively. For generality, we will often use $\mR^d$ for a $d$-dimensional Euclidean space; however, for the actually equation and computations we set $d = 3$ as we deal the usual differential operators on $\mR^3$.
\subsection{Wave packet decomposition.} For $(\bq,\bp) \in \mR^{2d}$, define $\phi_{\bq,\bp}^\epsilon$ as
\begin{equation}
\phi_{\bq,\bp}^\epsilon(\bx) = (-2\pi\epsilon)^{d/2}\exp\Big(\I\bp\cdot(\bx - \bq)/\epsilon - |\bx - \bq|^2/(2\epsilon)\Big)
\end{equation}
The Fourier–Bros–Iagolnitzer (FBI) transform on $\cS(\mR^d)$ is defined as;
\begin{align}
(\cF f)(\bq,\bp) &= \pi\epsilon^{-d/4}\ip{\psi^\epsilon_{\bq,\bp}}{f} \\
&= 2^{-d/2}(\pi\epsilon)^{-3d/4}\int_{\mR^d}\exp\Big(\I\bp\cdot(\bx - \bq)/\epsilon - |\bx - \bq|^2/(2\epsilon)\Big) f(\bx)\,\ud\bx
\end{align}
With inverse  transform $\big(\cF^\epsilon\big)^*$ defined on $\cS(\mR^{2d})$ given by
\begin{equation}
\Big(\big(\cF^\epsilon\big)^*f\big)(\bx) = 2^{-d/2}(\pi\epsilon)^{-3d/4}\int_{\mR^{2d}}\exp\Big(\I\bp\cdot(\bx - \bq)/\epsilon - |\bx - \bq|^2/(2\epsilon)\Big) g(\bq,\bp)\,\ud\bq\ud\bp.
\end{equation}
The following is a standard result from microlocal analysis, e.g. see ~\cite{bach2011introduction}.
\begin{proposition} For the Schwartz class, the FBI transform is an isometry on $\mR^d$, i.e., for any $f\in S(\mR^d)$
\begin{equation}
\|\cF^\epsilon f\|_{\rL^{2d}} = \|f\|_{\rL^{2d}}
\end{equation}
Furthermore; $\big(\cF^\epsilon\big)^*\cF^\epsilon = \Id_{\rL^2(\mR^d)}$. Meaning the domain of $\cF^\epsilon$ and $\big(\cF^\epsilon\big)^*$ can be extended to ${\rL^2(\mR^d)}$ and ${\rL^2(\mR^{2d})}$ respectively.
\end{proposition}
\begin{definition}
Let $\{u^\epsilon\} \subset\rL^2(\mR^d)$ be a family of functions that is uniformly bounded. Given $\delta > 0$, $\{\bu^\epsilon\}$ is \itshape{asymptotically high frequency with cut off $\delta$}\normalfont, if
\begin{equation}
\int_{\mR^{2d}\bs K_\delta} |(\cF^\epsilon u^{\epsilon})(\bq,\bp) |^2\ud\bq\ud\bp = \cO(\epsilon^\infty)
\end{equation}
as $\epsilon\to 0$.
\end{definition}
\begin{definition}
For $M_n \in \rL^\infty(\mR^{2d}; \mC^{N\times N})$ and a Schwartz function $\bu\in \cS(\mR^; \mC^N)$ for each $n = 1, \ldots, N$ define the Fourier integral operator $(\cI^\epsilon_n(t,M) \bu)(\bx)$ as
\begin{equation}
(\cI^\epsilon_n(t,M) \bu)(\bx) = (2\pi\epsilon)^{-3d/2}\int_{\mR^{3d}} G_n^\epsilon(t,\bx,\by,\bp,\bq)M(\bq,\bp)\bu(\by)\ud\bq\ud\bp\ud\by ,
\end{equation}
where
\begin{equation}
G_n^\epsilon(t,\bx,\by,\bp,\bq) = e^{\I \phi_n(t,\bx,\by,\bp,\bq)/\epsilon} 
\end{equation}
with phase function 
\begin{multline}\label{eq:phasefunc}
\phi_n(t,\bx,\by,\bp,\bq)  = \ds\frac{\I}{2}|\by-\bq|^2 - \bp\cdot(\by-\bq)\\
 + \ds\frac{\I}{2}|\bx - \bQ_n(t,\bq,\bp)|^2 + \bP_n(t,\bq,\bp)\cdot(\bx - \bQ_n(t,\bq,\bp)).
\end{multline}
\end{definition}

\begin{proposition}\label{eq:fio_bd}
If $M \in \rL^\infty(\mR^{2d}; \mC^{N\times N})$, for any $t$ and each $n = 1,\ldots, N$, $\cI^\epsilon_n(t,M)$  can be extended to a bounded linear operator on $\rL^2(\mR^d;\mC^N)$ with bound
\begin{equation}
\|\cI^\epsilon_n(t,M)\|_{\cL(\rL^2(\mR^d;\mC^N))} \leq 2^{-d/2} \|M_n\|_{\rL^\infty(\mR^d;\mC^N))}
\end{equation}
This is Proposition 3.7 in~\cite{LuYa:MMS}, a more general version is also proved~\cite{SwRo:09}, see Theorem 2.
\end{proposition}
\section{Elastic Wave equation}\label{sec:3}
Let $\bu:\mR^+\times \mR^3 \to \mR^3$ and $\rho, \lambda, \mu:\mR^3 \to \mR$, denote the operator $\cL\bu = (\lambda(\bx) + 2\mu(\bx))\nabla\big(\nabla \cdot \bu\big)  - \mu(\bx)\Curl\big(\Curl\bu\big)$, with the differential operators taken in the spacial variables. We are interested in the Cauchy problem with high frequency initial datum;
\begin{equation}\label{cp:eq:1}
\begin{cases}
(\rho(\bx)\partial_{t}^2 - \cL)\bu(t,\bx) = 0& \\
u(0,\bx) = \bu_0^\epsilon(\bx), &\bx \in \mR^d \\
\partial_tu(0,\bx) = \bu_1^\epsilon(\bx), &\bx \in \mR^d
\end{cases}
\end{equation}
We remark that the P-, S- wave speeds are given by;
\begin{equation}\label{eq:ps_speed}
c^2_{\tp}(\bx) = \frac{\lambda(\bx) + 2\mu(\bx)}{\rho(\bx)}\indent \text{ and }\indent c^2_{\ts}(\bx) = \frac{\mu(\bx)}{\rho(\bx)},
\end{equation}
respectively. With this formulation we must also have $\cO(\epsilon\bu_1^\epsilon) = \cO(\bu_0^\epsilon)$.  Define the following quantities
\begin{equation}\label{var:1}
\Theta(t,\bx) = \Div\bu(t,\bx),\indent \Psi(t,\bx) = \Curl\bu(t,\bx),\indent \bv(t,\bx) = \partial_t\bu(t,\bx).
\end{equation}
With $\bv = (v_1,v_2,v_3)^T$ and $\Psi = (\Psi_1,\Psi_2,\Psi_3)^T$, let $X = (v_1,v_2,v_3,\Theta,\Psi_1,\Psi_2,\Psi_3)^T$. Eq.~\eqref{cp:eq:1} can be written as a matrix system; in terms of the axillary variables~\eqref{var:1},
\begin{equation}\label{hyper_decomp}
\partial_t X = M_x\partial_x X + M_y\partial_y X + M_z\partial_z X.
\end{equation}
Using sparse notation; e.g. $M_{ij} = v$ is denoted $(i,j,v)$, the $M_x, M_y, M_z$ are as follows:
\begin{align}
M_x: (1,4,c_\tp^2), (2,7,c_\ts^2), (3,6,-c_\ts^2), (4,1,1), (7,2,1), (6,3,-1), \nonumber\\
M_y: (1,7,-c_\ts^2), (2,4,c_\tp^2), (3,5,c_\ts^2), (4,2,1), (5,3,1), (7,1,-1), \\
M_z: (1,6,c_\ts^2), (2,5,-c_\ts^2), (3,4,c_\tp^2), (4,3,1), (6,1,1), (5,2,-1). \nonumber
\end{align}
The eigenvalues of $M_x + M_y+ M_z$ are, $\pm c_\tp, 0 ,\pm c_\ts$, with $\pm c_\ts$ each having a multiplicity of 2. Once $\bv$ is known, solving for $\bu$ is an easy enough  task so it's not included in the matrix system.  We note here that, $\Theta$ represents the potential for the P-waves and $\Psi$ the potential of the S-waves. \br
Let $M$ denote the operator $M = M_x\partial_x + M_y\partial_y + M_z\partial_z$
and $X_0 = X(0)$ be the initial condition defined appropriately from the Cauchy problem in Eq.~\eqref{cp:eq:1}.
Denoting $\bp = (p_1, p_2, p_3)$, the eigenfunctions and left/right eigenvectors of the matrix system are as follows:
\begin{align}
H_0(\bq,\bp) &= 0 \label{eq:eig0}\\
R_0 &= (0,0,0,0, p_1, p_2,p_3)^\TT \nonumber  \\
L_0 &= \frac{1}{|p|^2}(0,0,0,0, p_1, p_2,p_3)^\TT \nonumber \\
H_{\tp\pm}(\bq,\bp) &= \pm c_\tp(\bq)|\bp|\\
R_{\tp\pm} &= (p_1,p_2,p_3,H_{\tp\pm}/c_\tp^2, 0, 0,0)^\TT, \nonumber \\
L_{\tp\pm} &= \frac{c_{\tp}^2}{2H_{\tp\pm}}(p_1,p_2,p_3,H_{\tp\pm}, 0,0,0)^\TT \nonumber \\
H_{\ts\pm}(\bq,\bp) &= \pm c_\ts(\bq)|\bp| \\
R_{\tp\pm,1} &= (-c_\ts(\bq)p_1p_2,c_\ts(\bq)(p_1^2 + p_3^2),c_\ts(\bq)p_2p_3,0, -p_3H_{\ts\pm}, 0,p_1H_{\ts\pm})^\TT,  \nonumber \\
L_{\ts\pm,1} &= \ds\frac{-1}{H_{\ts\pm}p_1p_2}(p_2^2+p_3^2,-p_1p_2,-p_1p_3,0,0,p_3H_{\ts\pm}, -p_2H_{\ts\pm})^\TT \nonumber\\
H_{\ts\pm}(\bq,\bp) &= \pm c_\ts(\bq)|\bp| \label{eq:eig1}\\
R_{\tp\pm,2} &= (c_\ts(\bq)p_1p_3,c_\ts(\bq)(p_2p_3),-c_\ts(\bq)(p_1^2 + p_2^2),0, -p_2H_{\ts\pm}, p_1H_{\ts\pm},0)^\TT,\nonumber \\
L_{\ts\pm,2} &= \ds\frac{-1}{H_{\ts\pm}p_2p_3}(p_2p_1,-p_1^2-p_3^2,p_3p_2,0,p_3H_{\ts\pm}, -p_1H_{\ts\pm})^\TT \nonumber
\end{align}
Assuming $p_1,p_2,p_3 \neq 0$, we normalize so that $L_m\cdot R_n = \delta_{n,m}$; otherwise, adjust the eigenvectors accordingly. Equations~\eqref{eq:eig0}-\eqref{eq:eig1} are used for the eigenvalue decomposition form of the FGA; they will not be used in the strong form for the evolution equation, we write them here for completeness. \\

The Hamiltonian associated with $\Theta$, $\Psi$
\begin{equation}
H_{\tp\pm,\ts\pm}(t,\bQ,\bP) = \pm c_{\tp,\ts}(\bQ_{\tp\pm,\ts\pm}(t,\bq,\bp)) \bP_{\tp\pm,\ts\pm}(t,\bq,\bp)
\end{equation}
The corresponding flows are given by
\begin{equation}\label{ew:flow}
\begin{cases}
\ds\ddt{\bQ_{\tp\pm,\ts\pm}}(0,\bq,\bp) &= \pm c_{\tp,\ts}(\bQ_{\tp\pm,\ts\pm}(t,\bq,\bp)) \ds\frac{\bP_{\tp\pm,\ts\pm}(t,\bq,\bp)}{|\bP_{\tp\pm,\ts\pm}(t,\bq,\bp)|} \\
\ds\ddt{\bP_{\tp\pm,\ts\pm}}(0,\bq,\bp) &= \ds\mp\partial_{\bQ}c_{\tp,\ts}(\bQ_{\tp\pm,\ts\pm}(t,\bq,\bp))|\bP_{\tp\pm,\ts\pm}(t,\bq,\bp)|
\end{cases}
\end{equation}
with initial conditions
\begin{equation}
\bQ_{\tp\pm,\ts\pm}(0,\bq,\bp) = \bq \text{ and } \bP_{\tp\pm,\ts\pm}(0,\bq,\bp) = \bp
\end{equation}
We remark that,
\begin{equation}\label{eq:hambound}
|\bp\cdot \partial_{\bq} H(t,\bq,\bp)| \lesssim |\bp|^2 \indent\text{ and } \indent|\bq\cdot \partial_{\bp} H(t,\bq,\bp)| \lesssim |\bq|^2
\end{equation}
so assumption A in~\cite{LuYa:MMS} is satisfied.
\begin{proposition} 
For $T > 0$ and $\delta > 0$, there is a constant $\delta_T$ such that
\begin{equation}
(\bQ_{\tp\pm,\ts\pm}(t,\bq,\bp), \bP_{\tp\pm,\ts\pm}(t,\bq,\bp)) \in K_T
\end{equation}
\end{proposition}
\begin{proof}
This is proposition 3.1 in~\cite{LuYa:MMS}, the proof requires the bound in eq.~\eqref{eq:hambound} and Gronwall's inequailty. 
\end{proof}
For the following definition we omit the branch subscript.
\begin{definition}
A map $\kappa_{\tp,\ts}:(\bd{q},\bd{p}) \rightarrow \bigl(\bd{Q}_{\tp,\ts}(\bd{q},\bd{p}),\bd{P}_{\tp,\ts}(\bd{q},\bd{p})\bigr)$ is called canonical transformation if the associated Jacobian matrix is symplectic, i.e., for any $(\bq,\bp)$
 \begin{equation*}
    J_{\tp,\ts}(\bd{q},\bd{p}) =
    \begin{pmatrix}
      (\partial_q \bd{Q}_{\tp,\ts})^{T}(\bd{q},\bd{p}) & (\partial_p
      \bd{Q}_{\tp,\ts})^{T}(\bd{q},\bd{p})
      \\
      (\partial_q \bd{P}_{\tp,\ts})^{T}(\bd{q},\bd{p}) & (\partial_p
      \bd{P}_{\tp,\ts})^{T}(\bd{q},\bd{p})
    \end{pmatrix},
  \end{equation*}
is symplectic, i.e., for any $(\bd{q},\bd{p})$,
   \begin{equation}\label{eq:symplectic}
    J_{\tp,\ts}^{T}
    \begin{pmatrix}
      0 & \Id_3 \\
      -\Id_3 & 0
    \end{pmatrix}
    J_{\tp,\ts} =
    \begin{pmatrix}
      0 & \Id_3 \\
      -\Id_3 & 0
    \end{pmatrix},
  \end{equation}
where $\Id_3$ is a $3\times 3$ identity matrix.
\end{definition}
\begin{proposition}
The map $\kappa_{\tp,\ts}$ is a canonical transform for any $T, \delta > 0$; furthermore it is bounded under sup norm.
\end{proposition}
\begin{proof}
This is proposition 3.4 in~\cite{LuYa:MMS}
\end{proof}
For a canonical transform $\kappa_{\tp,\ts}$ define the quantity $Z^{\kappa_{\tp,\ts},t}(\bq,\bp)$ for $|p| > 0$ as
\begin{equation}
Z^{\kappa_{\tp,\ts},t}  = \partial_{\bz}(\bQ(\bq,\bp) + \I \bP(\bq,\bp))
\end{equation}
With $\partial_{\bz} = (\partial_{\bq} - \I\partial_{\bp})$. Dropping the superscript $\kappa_{\tp,\ts}$,
\begin{equation}\label{eq:op_zZ_app}
Z = (\I \Id_3\ \ \Id_3) 
\left(\begin{array}{c c}
\partial_{\bq}\bQ & \partial_{\bq}\bP \\
\partial_{\bp}\bQ & \partial_{\bp}\bP \\
\end{array}\right)
\left(\begin{array}{c}
-\I \Id_3 \\
\Id_3
\end{array}\right)
\end{equation}
\begin{definition} 
The following notation will be useful. For $\ba \in C^\infty(\Omega,\mC)$, define for $k\in \mN$.
\begin{equation}
\Lambda_{k,\Omega}(\ba) := \max_{|\alpha_{\bp} | + |\alpha_{\bq} | = k} \sup_{(\bq,\bp)\in S}|\partial_{\bq}^{\alpha_{\bq}}\partial_{\bp}^{\alpha_{\bp}} \ba(\bq,\bp)|
\end{equation}
with $\alpha_{\bq}$, and $\alpha_{\bq}$ being multi-indices corresponding to $\bq$ and $\bp$ respectively.
\end{definition}
We will need the following lemma,
\begin{lemma}\label{lemma:1}
$Z^{\kappa_{\tp,\ts},t}$ is invertible for $(\bq,\bp) \in \mR^{2d}$ with $|\bp| > 0$.  Furthermore, for any $k \geq 0$ and $\delta > 0$, there exist constants $C_{k,\delta}$ such that
\begin{equation}
\Lambda_{k,K_{\delta}} \big((Z^{\kappa_{\tp,\ts},t}(\bq,\bp))^{-1}\big)\leq C_{k,\delta}
\end{equation}
\end{lemma}

\begin{proof}
The proof uses the property of they symplectic transform to bound the eigenvalues of $Z^{\kappa_{\tp,\ts},t}(Z^{\kappa_{\tp,\ts},t})^*$, see Lemma 5.1 in~\cite{LuYa:MMS} for details.
\end{proof}

\begin{lemma} For any vector $\bd{a}(\by,\bq,\bp) = (a_j)$ and matrix $M(\by,\bq,\bp) = (M_{ij})$ in Schwartz class, one has the following integration by parts formula in the componentwise form, with $\partial_{\bd{z}}=(\partial_{z_1},\partial_{z_2},\partial_{z_3})$,
\begin{equation}\label{lemma:2}
\begin{aligned}
a_j(x-Q)_j \sim&-\epsilon \partial_{z_m}\big(a_j Z_{jm}^{-1} \big), \\
(x-Q)_jM_{jl}(x-Q)_l \sim &\epsilon \partial_{z_m}Q_j M_{jl} Z_{lm}^{-1} + \cO(\epsilon^2),
\end{aligned}
\end{equation}
\end{lemma}
\begin{proof}
The proof requires integration by parts and invertibility of $Z$ from lemma~\eqref{lemma:1}. This is a special case of Lemma 5.2 in~\cite{LuYa:MMS} and it is Lemma 3.2 in ~\cite{LuYa:11} we refer to these for the detailed proof.
\end{proof}
\begin{theorem}
Given the Cauchy problem~\eqref{cp:eq:1} in terms of the matrix system~\eqref{hyper_decomp} with asymptotically high frequency initial condition $X_0^\epsilon$, the following estimate holds
\begin{equation}\label{eq:fga:efd}
\|X - X_{\FGA,0} \|_{\rL^2} \lesssim \epsilon \|X_0^\epsilon\|_{\rL^2}
\end{equation}
\end{theorem}
\begin{proof}
This is the main content of~\cite{LuYa:CPAM}, $X_{\FGA}$ is defined as
\begin{multline}\label{eq:xfga}
X_{\FGA,0} = \left(2\pi\epsilon\right)^{-3d/2}\sum_{n=1}^7\int_{\mR^{3d}} \sigma_n(t,\bq,\bp) R_n(t,\bQ_n,\bP_n)L_n(t,\bq,\bp)^T \\
\times G_n^\epsilon(t,\bx,\by,\bq,\bp)X^\epsilon_0(\by)\ud\by\ud\bq\ud\bp.
\end{multline}
sub-scripting over $n$ (instead of $\tp\pm,\ts\pm,0$) for the eigenfunctions and left and right eigenvectors defined from Eqs.~\eqref{eq:eig0}-\eqref{eq:eig1} with $\bQ_n = \bQ_n(t,\bq,\bp)$, $\bP_n = \bP_n(t,\bq,\bp)$ and $\sigma_n(t,\bq,\bp)$ solving the evolution equation
\begin{equation}
\begin{cases}
\ds\ddt{} \sigma_n(t,\bq,\bp) + \sigma_n(t,\bq,\bp)\lambda_n(t,\bq,\bp) = 0&  \\
\sigma_n(0,\bq,\bp) = 2^{d/2}
\end{cases}
\end{equation}
where
\begin{align}
\lambda_n &= L^T\big(\partial_{P_n}H_n\cdot\partial_{Q_n}R_n - \partial_{Q_n}H_n\cdot \partial_{P_n}R_n\big) \\
& - (\partial_{z_k}L_n)^T\Big( M_j - \partial_{P_n,j}H_n + \I (\partial_{Q_n,j}H_n - P_{n,l}\partial_{Q_{n,j}}M_{l})\Big)R_nZ^{-1}_{n,jk}  \nonumber\\
& + \partial_{z_s}Q_{n,j}Z_{n,ks}^{-1}L^T\Big(-\partial_{Q_{m,j}}M_k + \frac{\I}{2}P_{n,l}\partial^2_{{Q_n,j},{Q_n,k}}M_l\Big)R_n \nonumber
\end{align}
With $Q_n, P_n$ evaluated at $(t,\bq,\bp)$ and $M_j$ evaluated at $\bQ_n$ and $H_n, L_n, R_n$ evaluated at $(\bQ_n,\bP_n)$. 
\end{proof}
We note that the work done in~\cite{LuYa:CPAM} for the first ordered FGA is the same as the operator is diagonalizable and the eigenspaces are non degenerate. We write $X_\FGA$, to show first order asymptotic equivalence; see corollary~\ref{cor:1}, to our FGA using the amplitude factor derived from projections.
\section{The FGA via projection}\label{sec:4}
We begin by defining the FGA derived from its strong form of the evolution equation; define the first order FGA as:
\begin{multline}\label{eq:fga_sol}
 u_{\FGA,0}(t,\bx) =  (2\pi\epsilon)^{-3d/2}\int_{\mR^{3d}} \sum_{b=\pm}\ba_{\tp,b,0}(t,\by,\bq,\bp) G_{\tp,b}^{\epsilon}(t,\bx,\by,\bq,\bp) \\
 + (\ba_{\tsh,b,0}(t,\by,\bq,\bp) + \ba_{\tsv,b,0}(t,\by,\bq,\bp))G_{\ts,b}^{\epsilon}(t,\bx,\by,\bq,\bp) \ud \by\ud \bp \ud\bq
\end{multline}
Often times the computations for the branches and P-,S- wavefields will be similar, in these cases we either omit the subscript or subscript by $n$ instead of $(\tp\pm,\tsh\pm,\tsh\pm)$ or $(\tp\pm,\ts\pm)$.
With this notation we can define eq.~\eqref{eq:fga_sol} in more compactly as;
\begin{equation}\label{eq:fga_sol:n}
 u_{\FGA,0}(t,\bx) =  (2\pi\epsilon)^{-3d/2}\int_{\mR^{3d}} \sum_{n=1}^6\ba_{n,0}(t,\by,\bq,\bp) G_{n}^{\epsilon}(t,\bx,\by,\bq,\bp) \ud \by\ud \bp \ud\bq
\end{equation}
For $k > 1$, define the k-th ordered FGA with a correction term as;
\begin{multline}
u_{\FGA,k}(t,\bx)  =  u_{\FGA,0}(t,\bx) \\
 +  (2\pi\epsilon)^{-3d/2}\int_{\mR^{3d}} \sum_{j=1,n=1}^{k,6} \epsilon^j\big(\ba_{n,j}(t,\by,\bq,\bp) + \ba_{n,j}^\perp(t,\by,\bq,\bp)\big)\\
	\times G_n^\epsilon(t,\bx,\by,\bq,\bp) \ud \by\ud \bp \ud\bq
\end{multline}
the terms $\ba_{n,1}^\perp$ will be defined later in this section~\ref{4.2}.
We define a standard smooth cutoff function $\chi_\delta: \mR^{2d} \to [0,1]$ for the set $K_\delta$ as
\begin{equation}
\chi_\delta(\bq,\bp) = 
\begin{cases}
1, &(\bq,\bp) \in K_\delta \\
0, &(\bq,\bp) \in \mR^{}\bs K_{\delta/2}
\end{cases}
\end{equation}
and for any $k\in\mN$, there exists a constant $C_{K,\delta}$ such that
\begin{equation}
\Lambda_{k}\big(\chi_\delta(\bq,\bp)\big) < C_{K,\delta}
\end{equation}
We define the filtered version of the FGA as follows;
\begin{multline}
 \tilde{u}_{\FGA,k}(t,\bx) =  (2\pi\epsilon)^{-3d/2}\int_{\mR^{3d}} \chi_\delta \sum_{j=0,n=1}^{k,6} \epsilon^j\big(\ba_{n,j}(t,\by,\bq,\bp) + \ba_{n,j}^\perp(t,\by,\bq,\bp)\big) \\
\times G_n^\epsilon(t,\bx,\by,\bq,\bp) \ud \by\ud \bp \ud\bq 
\end{multline}
where $\ba_{n,0}^\perp = 0$. Define the unit vectors $\hat{\bd{N}}_{\tp\pm},\hat{\bd{N}}_{\tsv\pm},\hat{\bd{N}}_{\tsh\pm}$ that points in the direction P, SV, or SH waves respectively.  Then $\ba_{n}(t,\by,\bq,\bp)$ is defined as follows,
\begin{equation}
\ba_{n}(t,\by,\bq,\bp) = a_{n}(t,\bq,\bp)\alpha_{n}^\epsilon(\by,\bq,\bp)\hat{\bd{N}}_n(t,\bq,\bp).
\end{equation}
Where $\hat{\bd{N}}_n(0,\bq,\bp) = \hat{\bd{n}}_{n}$ and $\alpha^\epsilon$ incorporates the initial conditions,
\begin{align}\label{eq:alpha}
\alpha^\epsilon_{n}(\bd{y},\bd{q},\bd{p})=\frac{1}{2c_{n}\abs{\bp}^3}\Bigl(\bd{u}_0^\epsilon(\bd{y})c_{n}\abs{\bd{p}}\pm
  \I\epsilon\bd{u}_1^\epsilon(\bd{y}) \Bigr)\cdot\hat{\bd{n}}_{n}.
\end{align}
The scalar functions $a_n$ satisfying the following evolution equations;
\begin{align}\label{eq:amp_P}
&\frac{\ud a_{\tp}}{\ud t} =  a_{\tp}\biggl(\pm\frac{\partial_{\bd{Q}_{\tp}}c_{\tp}\cdot \bd{P}_{\tp}}{|\bd{P}_{\tp}|} + \frac{1}{2}\tr\Bigl(Z_{\tp}^{-1}\frac{\ud Z_{\tp}}{\ud t}\Bigr)\biggr), \\
& \frac{\ud a_{\tsv}}{\ud t} =  a_{\tsv}\biggl(\pm\frac{\partial_{\bd{Q}_{\ts}}c_{\ts}\cdot \bd{P}_{\ts}}{|\bd{P}_{\ts}|} + \frac{1}{2}\tr\Bigl(Z_{\ts}^{-1}\frac{\ud Z_{\ts}}{\ud t}\Bigr)\biggr)-a_{\tsh}\frac{\ud \hat{\bd{N}}_{\tsh}}{\ud t}\cdot\hat{\bd{N}}_\tsv, \label{eq:amp_SV}\\
& \frac{\ud a_{\tsh}}{\ud t} = a_{\tsh}\biggl(\pm\frac{\partial_{\bd{Q}_{\ts}}c_{\ts}\cdot \bd{P}_{\ts}}{|\bd{P}_{\ts}|} + \frac{1}{2}\tr\Bigl(Z_{\ts}^{-1}\frac{\ud Z_{\ts}}{\ud t}\Bigr)\biggr)+a_{\tsv}\frac{\ud \hat{\bd{N}}_{\tsh}}{\ud t}\cdot\hat{\bd{N}}_\tsv, \label{eq:amp_SH}
\end{align}
Eq.~\eqref{eq:alpha} is derived from the writing $\bu_0(\bx)$, $\bu_1(\bx)$ in terms of FBI transform and it's inverse, i.e.
\begin{align}
\bu_0(\bx) &=  (2\pi\epsilon)^{-3d/2}\int_{\mR^{3d}} \bu_0(\by) G^\epsilon(0,\bx,\by,\bq,\bp)\ud \by\ud \bp \ud\bq
\end{align}
and decomposing the integrand interms of the basis $\{ \hat{n}_{\tp},\hat{n}_{\tsv},\hat{n}_{\tsh}\}$.  
\begin{remark}~\label{rmk:1}
It is easy to check that $\bu_\FGA(0,\bx) = \bu_0(\bx)$ as $\cF^*(\cF(\bu_0)) = \bu_0$
\end{remark}
\subsection{Derivation of the evolution equation}
\noindent For the general computation we will drop the subscript and branch notation.  The calculations for the two branches are identical for both P,S waves. First we introduce the notation, $f\sim g$ if
\begin{equation}\label{eq:sim}
\int_{\mR^{3d}} f(y)G^\epsilon(t,\bx,\by,\bq,\bp)\ud\by\ud\bq\ud\by = \int_{\mR^{3d}} g(y)G^\epsilon(t,\bx,\by,\bq,\bp)\ud\by\ud\bq\ud\by.
\end{equation}
Note that eq.~\eqref{cp:eq:1} can be written as follow
\begin{equation}\label{eq:ewe_cpcs}
\partial^2_t\bu = c_\tp^2(\bx)\nabla(\nabla\cdot \bu) - c_\ts^2(\bx)\nabla\times\nabla \times \bu,
\end{equation}
with $c_\tp$ and $c_\ts$ given in~\eqref{eq:ps_speed}. For simplicity in the asymptotics, we assume $\mu$ and $\lambda$ are constant and write~\eqref{eq:ewe_cpcs} as
\begin{equation}\label{eq:ewe_curl}
\rho(\bd{x})\partial^2_t\bu = (\lambda + 2\mu)\nabla(\nabla\cdot \bu) - \mu\nabla\times\nabla \times \bu.
\end{equation}
Eq.~\eqref{eq:ewe_curl} is linear, and thus one can derive the prefactor equations for P- and S-waves individually by assuming $\bd{A}_\tp\parallel \bd{P}$ or $\bd{A}_\ts\perp \bd{P}$, with the following Gaussian wave packet solution ansatz;
\begin{equation}\label{solform:1}
\bu_{\tp,\ts}(t,\bd{x},\bd{y},\bd{q},\bd{p}) = \bd{A}_{\tp,\ts}(t,\bd{q},\bd{p})G^\epsilon_{\tp,\ts}(t,\bd{x},\bd{y},\bd{q},\bd{p}).
\end{equation} 
Without loss of generality, we first consider the prefactor equation for the P-wave, with the governing equations for $\bd{Q}_\tp$ and $\bd{P}_\tp$ given by eq.~\eqref{ew:flow}.
\noindent Plugging eq.~\eqref{solform:1} into eq.~\eqref{eq:ewe_curl} and expanding the asymptotics in the weak sense of \eqref{eq:sim} yield
\begin{multline}\label{weq:1}
\rho \left(\bd{A}_{tt} +  2\I k\bd{A}_t\Phi_t +  \I k\bd{A} \Phi_{tt}- k^2 \bd{A} \Phi_t^2  \right)
\sim   (\lambda + 2\mu)\left(\I k\nabla(\bd{A}\cdot \nabla \Phi) -  k^2(\bd{A} \cdot\nabla\Phi)\nabla\Phi \right)  \\ -\mu\left(\I k  \Curl(\nabla \Phi \times \bd{A}) - k^2 \nabla \Phi\times(\nabla \Phi\times \bd{A})\right). \end{multline}
The spatial and temporal derivatives of $\Phi$ are given by
\begin{equation}\label{eq:Phi_der}
\begin{aligned}
&\nabla \Phi =  \bP + \I (\bd{x} - \bQ),\ \Delta \Phi  =  3\I,\ \nabla^2\Phi = \I\Id_3, \\
&|\nabla\Phi|^2 =  |\bP|^2 + 2\I \bP\cdot (\bd{x} - \bQ) - |\bd{x}-\bQ|^2, \\
&\Phi_t =\ (\bP_t- \I \bQ_t)\cdot (\bd{x}-\bQ) - \bP\cdot \bQ_t, \\
&\Phi_t^2 =\ \big[(\bP_t- \I \bQ_t)\cdot (\bd{x}-\bQ)\big]^2 + (\bP\cdot \bQ_t)^2 - 2(\bP\cdot \bQ_t)\big[(\bP_t- \I \bQ_t)\cdot (\bd{x}-\bQ)\big],  \\
&\Phi_{tt} =\ (\bP_{tt}- \I \bQ_{tt})\cdot (\bd{x}-\bQ) - (\bP_t- \I \bQ_t)\cdot \bQ_t - \bP_t\cdot \bQ_t -  \bP\cdot \bQ_{tt}.
\end{aligned}
\end{equation}
Notice that the terms containing $k(\bd{x}-\bd{Q})$ will be of $\cO(1)$ by the lemma of integration by parts, and for P-waves, $\bd{P}\times \bd{A} = 0$ and $\Curl\bigl((\bd{x}-\bQ)\times {\bd{A}}\bigr) = -2\bd{A}$. Plugging the derivatives of $\Phi$ in eq.~\eqref{eq:Phi_der} into eq.~\eqref{weq:1} produces, after neglecting the $\cO(1)$ and lower order terms,
\begin{equation}\label{FGA:1}
\begin{aligned}
2 k \rho \bd{A}_t\big(\bP\cdot \bQ_t \big) \sim &\  k\rho \bd{A} \Big( -(\bP_t- \I \bQ_t)\cdot \bQ_t - \bP_t\cdot \bQ_t -  \bP\cdot \bQ_{tt} \Big) \\
+ &\ \I k \rho \bd{A}\Big( (\bd{x}-\bQ)\cdot\Big((\bP_t- \I \bQ_t)\otimes(\bP_t- \I \bQ_t)\Big)(\bd{x}-\bQ) \Big)  \\
+ &\ \I k^2 \rho \bd{A} \Big(-2(\bP\cdot \bQ_t)\big((\bP_t- \I \bQ_t)\cdot (\bd{x}-\bQ)\big) + (\bP\cdot \bQ_t)^2 \Big)  \\
+ &\ (\lambda + 2\mu)\Big(-  \I k \bd{A} + k^2\big((\bd{A}\cdot \bP)(\bd{x}-\bQ)  \\
 &\ +(\bd{A}\otimes \bP)(\bd{x}-\bQ)\big) + \I\big((\bd{x}-\bQ)\otimes(\bd{x}-\bQ)\big)\bd{A} \Big) \\
-&\ 2\I \mu k \bd{A} - \I\mu k^2\Big(((\bd{x}-\bQ)\cdot \bd{A})(\bd{x}-\bQ) - |\bd{x}-\bQ|^2\bd{A}\Big) \\
-&\  \mu k^2\Big((\bP\cdot \bd{A})(\bd{x}-\bQ) - (\bP\cdot(\bd{x}-\bQ))\bd{A}\Big),
\end{aligned}
\end{equation}
where $\otimes$ means the tensor product, e.g., $(\bd{A}\otimes \bd{P})_{jl}=A_jP_l$.

\noindent Expanding $\rho(\bd{x})$ around $\bQ$ and truncating at order third order,
\begin{align}
\rho(\bd{x}) =& \ \rho + \partial_{\bQ}\rho \cdot(\bd{x}-\bQ)+ \ds\frac{1}{2}(\bd{x}-\bQ)\cdot \partial_{\bQ\bQ}^2\rho(\bd{x}-\bQ),
\end{align}
and noticing that $\rho(\bd{Q}) c^2(\bd{Q})=\lambda+2\mu$ is constant, one has
\begin{align}\label{density:1}
\partial_{\bQ}\rho c^2 + 2c \partial_{\bQ}c\rho = 0,\quad \partial_{\bQ}\rho = -\ds\frac{2c \partial_{\bQ}c\rho}{c^2}.
\end{align}
Taking the second derivative for $\partial_{\bQ\bQ}\rho$ and substituting eq.~\eqref{density:1} bring
\begin{equation}\label{eq:rho_taylor}
\rho(\bd{x}) =p  \rho -\ds\frac{2\partial_{\bQ}c\rho}{c}\cdot(\bd{x}-\bQ)  + (\bd{x}-\bQ)\cdot\left(3\ds\frac{\partial_{\bQ}c\partial_{\bQ}c^T\rho}{c^2}  - \ds\frac{\partial_{\bQ\bQ}c\rho}{c}\right)(\bd{x}-\bQ).
\end{equation}
Plugging eqs,~\eqref{ew:flow} and \eqref{eq:rho_taylor} into eq.~\eqref{FGA:1}, and dividing by $\rho$ yield, in componentwise form,
\begin{equation}\label{FGA:4}
\begin{aligned}
2 k \partial_tA_i c|\bP| \sim & k A_i \Big(c (\partial_{\bQ}c)_j P_j +\I c^2\Big) - \ds \I c^2 k A_i  \\
&\ +k^2 A_i(x-Q)_jM_{jl}(x-Q)_l  - 2k^2A_ic^2P_j (x-Q)_j + \ds c^2N_{i} \\
 &\ - \ds\frac{2\I k\mu}{\rho}A_i - \ds\frac{k^2\mu}{\rho}\Big(P_jA_j(x-Q)_i - P_j(x-Q)_j A_i\Big)\\
&\ -\ds\frac{\I k^2\mu}{\rho}\Big((x-Q)_jA_j(x-Q)_i - |\bd{x}-\bQ|^2A_i\Big),
\end{aligned}
\end{equation}
where $M_{jl}$ and $N_{i}$ are given as follows,
\begin{equation}\label{fga_fact:1}
\begin{aligned}
M_{jl} =&\  -\I|\bP|^2c(\partial^2_{\bQ}c)_{jk} + c\big(P_l +P_j\big)(\partial_{\bQ}c)_j-\I c^2\ds\frac{P_jP_l}{|\bP|^2},  \\
N_{i} = &\ A_j P_j(x-Q)_i + A_iP_j(x-Q)_j +\I(x-Q)_i(x-Q)_jA_j.
\end{aligned}
\end{equation}

\noindent Assuming that $\bd{A} = a_{\tp}\hat{\bd{N}}_{\tp}$, with $\displaystyle \hat{\bd{N}}_\tp=\frac{\bd{P}}{\abs{\bd{P}}}$, then \begin{equation}
\partial_t \bd{A} = \partial_ta_\tp\hat{\bd{N}}_{\tp} + a_{\tp}\partial_t\hat{\bd{N}}_{\tp}=\partial_ta_\tp \frac{\bd{P}}{\abs{\bd{P}}} + a_{\tp}\partial_t\Bigl(\frac{\bd{P}}{\abs{\bd{P}}}\Bigr).
\end{equation}
Plugging this into eq.~\eqref{FGA:4} with the ray equations \eqref{ew:flow}, using the P-wave velocity eq.~\eqref{eq:ps_speed} and grouping in powers of $(\bd{x}-\bQ)$ produce
\begin{align}
2k\partial_t a_{\tp}c\abs{\bd{P}}P_i \sim &\ -2ka_{\tp}c|\bP|^2 \partial_t \Bigl(\frac{P_i}{\abs{\bd{P}}}\Bigr) + k a_{\tp} c (\partial_{\bQ}c)_jP_j P_i  \nonumber\\
&\  -\ds\frac{2\I k\mu P_i }{\rho}a_{\tp} + k^2a_{\tp}|\bP|^2\left(c^2 - \ds\frac{\mu}{\rho} \right)(x-Q)_i \\
 &\ + k^2 a_{\tp} \left(\ds\frac{\mu}{\rho}-c^2 \right)P_jP_i(x-Q)_j  	+ k^2\ds\frac{\I\mu a_{\tp}}{\rho}|x-Q|^2P_i \nonumber \\
 &\ + k^2\I a_{\tp}P_j\left(c^2 - \ds\frac{\mu}{\rho}\right)(x-Q)_j(x-Q)_i + k^2a_{\tp}P_i(x-Q)_j M_{jl}(x-Q)_l. \nonumber
\end{align}
Denoting $\partial_{\bd{z}}=(\partial_1,\partial_2,\partial_3)$ for an ease of notations, applying eq.~\eqref{lemma:1} and dropping the lower order terms yield
\begin{equation}\label{eq:1}
\begin{aligned}
2\partial_ta_{\tp}c\abs{\bd{P}}P_i \sim &\ -2ka_{\tp}c|\bP|^2 \partial_t \Bigl(\frac{P_i}{\abs{\bd{P}}}\Bigr) + a_{\tp}c (\partial_{\bQ}c)_jP_j P_i - \ds\frac{2\I\mu a_{\tp}}{\rho}P_i \\
- &\ 
\partial_{l}\left(\alpha\Big(|\bP|^2 - P_jP_i\Big)a_{n}\left(\ds\frac{\mu}{\rho} - c^2\right) Z^{-1}_{jl}\right)\\
+ &\ 
\I a_{\tp}P_j\left(c^2 - \ds\frac{\mu}{\rho}\right) Z_{il}^{-1}\partial_{l}Q_j
+ \ds\frac{\I\mu a_{\tp}}{\rho}P_i Z_{jl}^{-1}\partial_{l}Q_j \\
+ &\ 
a_{\tp}P_i\left(c(\partial_{\bQ}c)_jP_r  +cP_j(\partial_{\bQ}c)_j -\I c^2\ds\frac{P_jP_r}{|\bP|^2}\right)cZ_{rl}^{-1}\partial_{l}Q_j \\
 - &\ a_{\tp}\I P_i|\bP|^2c(\partial^2_{\bQ}c)_{jr} Z_{rl}^{-1}\partial_{l}Q_j.
\end{aligned}
\end{equation}
To derive an ordinary differential equation (ODE) instead of a partial differential equation (PDE) for $a_{\tp}$, one needs to simplify the terms containing $\partial_ka_{\tp}$ as
\begin{equation}\label{eq:a_der}
\Big(|\bP|^2 - P_jP_i\Big)\partial_l\left(a_{\tp}\left(\ds\frac{\mu}{\rho} - c^2\right)
 Z^{-1}_{jl}\right).
\end{equation}
Recall that eq.~\eqref{eq:1} holds in the sense of integral form \eqref{eq:sim}, and now we shall consider a strong form of eq.~\eqref{eq:1}, i.e., equate the integrands of the integrals on both sides. After taking the dot product of integrands with $\bP$, one has
\begin{align*}
 2\ds\frac{\partial_t a_{\tp}}{a_{\tp}}
= &\ 
\ds2\frac{(\partial_{\bQ}c)_jP_j}{|\bP|} - \ds\frac{2\mu\I}{\rho c|\bP|} + \ds\frac{\I\mu}{\rho c|\bP|}\left(\delta_{ij} - \ds\frac{P_jP_i}{|\bP|^2}\right) Z_{il}^{-1}\partial_{l}Q_j \\
+ &\ 
\ds\frac{1}{c}\left(c^2 - \ds\frac{\mu}{\rho}\right)\ds\frac{P_i}{|\bP|}\partial_{l}\left(\ds\frac{P_jP_i}{|\bP|^2}\right) Z_{jl}^{-1} + \ds\frac{1}{|\bP|}(\partial_{\bQ}c)_jP_k Z_{kl}^{-1}\partial_{l}Q_j  \\
+&\ \ds\frac{1}{|\bP|}P_j(\partial_{\bQ}c)_j Z_{kl}^{-1}\partial_{l}Q_j -\I|\bP|(\partial^2_{\bQ}c)_{jk} Z_{kl}^{-1}\partial_{l}Q_j,
\end{align*}
where the terms \eqref{eq:a_der} actually become zero since $\bP\cdot(|\bP|^2 - \bP\otimes\bP)=0$, and we have used the fact that
\begin{equation*}
\frac{P_i}{\abs{\bd{P}}}\partial_t \Bigl(\frac{P_i}{\abs{\bd{P}}}\Bigr)=\frac{1}{2}\biggl( \frac{P_i}{\abs{\bd{P}}}\partial_t \Bigl(\frac{P_i}{\abs{\bd{P}}}\Bigr) + \partial_t \Bigl(\frac{P_i}{\abs{\bd{P}}}\Bigr) \frac{P_i}{\abs{\bd{P}}}\biggr)=\partial_t\biggl(\frac{P_i}{\abs{\bd{P}}}\frac{P_i}{\abs{\bd{P}}}\biggr)=0,
\end{equation*}
which implies $\displaystyle \partial_t \Bigl(\frac{P_i}{\abs{\bd{P}}}\Bigr)P_i=0$.

\noindent Since $ Z = \partial_{\bz}(\bQ + \I \bP)$ by eq.~\eqref{eq:op_zZ_app}, $\partial_t  Z=\partial_t\partial_{\bz}\bQ + \I\partial_t\partial_{\bz}\bP$. Then eq.~\eqref{ew:flow} implies
\begin{equation}\label{eq:DzDt}
\begin{aligned}
\partial_t\partial_{\bz}\bQ
=&\ \partial_{\bz}\bQ \ds\frac{\partial_{\bQ}c\otimes \bP}{|\bP|} + c\partial_{\bz}\bP \left(\frac{\Id_3}{|\bP|} - \ds\frac{\bP \otimes \bP}{|\bP|^3}\right), \\
\partial_t\partial_{\bz}\bP
=&\ -|\bP|\partial_{\bz}\bQ \partial^2_{\bQ}c - \partial_{\bz}\bP\ds\frac{\bP\otimes \partial_{\bQ} c}{|\bP|}.
\end{aligned}
\end{equation}
Using eq.~\eqref{eq:DzDt} for further simplifications give
\begin{equation}
\begin{aligned}
2\ds\frac{\partial_t a_{\tp}}{a_{\tp}}
= &\ %
 \ds 2\frac{(\partial_{\bQ} c)_iP_i}{|\bP|} + \tr\left( Z^{-1} \partial_t Z\right) \\
 &\ + %
\ds\frac{1}{c}\left(c^2 - \ds\frac{\mu}{\rho}\right)\ds\frac{P_i}{|\bP|}\partial_{l}\left(\ds\frac{P_jP_i}{|\bP|^2}\right) Z_{jl}^{-1} - \ds\frac{1}{ c|\bP|}\left(c^2 - \ds\frac{\mu}{\rho} \right)\left(\delta_{ij} - \ds\frac{P_jP_i}{|\bP|^2}\right) Z_{il}^{-1}\partial_{l}P_j,
\end{aligned}
\end{equation}
where the last two terms can be grouped as
\begin{align}
\ds\frac{1}{c|\bP|}\left(c^2 - \ds\frac{\mu}{\rho}\right)\partial_{l}\left(P_i\Big(\ds\frac{P_jP_i}{|\bP|^2}-\delta_{ij}\Big)\right) Z_{jl}^{-1} =0,
\end{align}
which implies a clean ODE for $a_{\tp}$ as in eq.~\eqref{eq:amp_P},
\begin{equation}
\frac{\ud a_{\tp}}{\ud t} = a_\tp\left(\ds\frac{\partial_{\bQ}c\cdot \bP}{|\bP|} + \ds\frac{1}{2}\tr\left( Z^{-1} \frac{\ud Z}{\ud t}\right)\right).
\end{equation}
For the $\cO(1)$ equation, from the expansion and similar steps we arrive at, in component form;
\begin{eqnarray}
2c|P|\ddt{a_1}& = &2a_1c\partial c_j P_j |\bP| + a_1c|P|\big(Z^{-1}_{jk}\ddt{Z_{jk}}\big)  \\
 &- &\imath(a_{0})_{tt} - \imath a_0
				-\partial_k\left(4\Big(a_{0}\frac{P_i}{|P|}\Big)_tP_i\partial c_j \cZ^{-1}_{jk} \right) \frac{P_i}{|P|} \\
&-& 2(a_0)_t|P|\left(
				          3\ds\frac{\partial c_j\partial c_k}{c}  
				         - \partial^2c_{jk}
						\right)\cZ^{-1}_{jl}\partial_l(Q_k) \nonumber \\
& +& \partial_k\left(a_0\frac{P_i}{|P|}M_{jk}\cZ^{-1}_{jk}\right)\frac{P_i}{|P|} + a_0 N_{jk} \cZ^{-1}_{jl}\partial_l(Q_k)\nonumber\\
& + & 2\imath\partial_n\left(a_{0}\frac{P_i}{|P|}\Big(
							\frac{\partial c_l\partial c_j\partial c_k|P|^2}{c} 
							+ 2\imath \partial c_l\partial c_k P_j 
							- c\partial c_l \ds\frac{P_jP_k}{|P|^2}
					\Big)\cZ^{-1}_{jm}\partial_m(Q_k)\cZ^{-1}_{ln}\right)\frac{P_i}{|P|} \nonumber
\end{eqnarray}
Where 
\begin{equation}
M_{jk} = 2\partial c_j|P| 
					    + 2\imath c\ds\frac{P_j}{|P|}
								+\partial c_j \partial c_kP_k 
							  + c \partial^2c_{kj}P_k
							  +\imath\ds\frac{2c \partial c_k P_kP_j}{|P|^2} 
							  +\imath c \partial c_j
\end{equation}
\begin{multline}
N_{jk} = 4\frac{\partial c_j\partial c_k}{c}|P| 
							+ 4\imath \ds\frac{\partial c_kP_j}{|P|}
							- \partial^2c_{jk}\partial c_nP_n
							- \ds \imath c \partial^2c_{jk}\\
								+ \ds\frac{\partial c_j\partial c_k\partial c_nP_n}{c} 
								+2\partial c_j\partial^2c_{nk}P_n
								+ 4\imath\ds\frac{\partial c_j \partial c_n P_nP_k}{|P|^2} 
								+\imath \partial c_j \partial c_k
\end{multline}
We write this as
\begin{equation}
\frac{\ud a_{\tp,1}}{\ud t} = a_{\tp,1}\left(\ds\frac{\partial_{\bQ}c\cdot \bP}{|\bP|} + \ds\frac{1}{2}\tr\left( Z^{-1} \frac{\ud Z}{\ud t}\right)\right) + F_p(a_{\tp,0},\bQ,\bP,c_\tp).
\end{equation}
with $F_p$ containing first and second derivatives of its arguments, which are smooth for $|\bP| > 0$. 
\noindent Similarly, one can derive the prefactor equations for SV- and SH-waves by assuming $\bd{A}=a_{\tsv}\hat{\bd{N}}_{\tsv}+a_{\tsh}\hat{\bd{N}}_{\tsh}$ with $\hat{\bd{N}}_\tsv\perp \bd{P}$, $\hat{\bd{N}}_\tsh\perp \bd{P}$ and $\hat{\bd{N}}_\tsv\perp \hat{\bd{N}}_\tsh$ in eq.~\eqref{solform:1}. The calculations will be essentially the same as the prefactor equation for P-waves except that one will have the diabatic coupling terms of $\hat{\bd{N}}_{\tsv}$ and $\hat{\bd{N}}_{\tsh}$ as shown below,
\begin{equation}\label{eq:2}
\begin{aligned}
& \frac{\ud a_{\tsv}}{\ud t} = a_{\tsv}\biggl(\frac{\partial_{\bd{Q}_{\ts}}c_{\ts}\cdot \bd{P}_{\ts}}{|\bd{P}_{\ts}|} + \frac{1}{2}\tr\Bigl(Z_{\ts}^{-1}\frac{\ud Z_{\ts}}{\ud t}\Bigr)\biggr)-a_{\tsh}\biggl(\frac{\ud \hat{\bd{N}}_{\tsh}}{\ud t}\cdot\hat{\bd{N}}_\tsv+m_{\tsh\rightarrow\tsv}\biggr), \\
& \frac{\ud a_{\tsh}}{\ud t} = a_{\tsh}\biggl(\frac{\partial_{\bd{Q}_{\ts}}c_{\ts}\cdot \bd{P}_{\ts}}{|\bd{P}_{\ts}|} + \frac{1}{2}\tr\Bigl(Z_{\ts}^{-1}\frac{\ud Z_{\ts}}{\ud t}\Bigr)\biggr)-a_{\tsv}\biggl(\frac{\ud \hat{\bd{N}}_{\tsv}}{\ud t}\cdot\hat{\bd{N}}_\tsh+m_{\tsv\rightarrow\tsh}\biggr),
\end{aligned}
\end{equation}
where the interaction terms are given by
\begin{equation*}
\begin{aligned}
&m_{\tsh\rightarrow\tsv} = \I \frac{\lambda+\mu}{\rho c_\ts\abs{P_{\ts}}}\Bigl(\hat{\bd{N}}_{\tsv}\cdot(Z_\ts^{-1}\partial_{\bd{z}}\bd{Q}_s)\hat{\bd{N}}_{\tsh} - \hat{\bd{N}}_{\tsh}\cdot(Z_\ts^{-1}\partial_{\bd{z}}\bd{Q}_s)\hat{\bd{N}}_{\tsv} \Bigr), \\ & m_{\tsv\rightarrow\tsh}=-m_{\tsh\rightarrow\tsv}.
\end{aligned}
\end{equation*}
Also, note that by $\hat{\bd{N}}_\tsv\perp \hat{\bd{N}}_\tsh$, one has that $\displaystyle \frac{\ud \hat{\bd{N}}_{\tsh}}{\ud t}\cdot\hat{\bd{N}}_\tsv + \frac{\ud \hat{\bd{N}}_{\tsv}}{\ud t}\cdot\hat{\bd{N}}_\tsh =0$.

\noindent Next, we shall show that $m_{\tsh\rightarrow\tsv}=m_{\tsv\rightarrow\tsh}=0$ by proving that $Z_\ts^{-1}\partial_{\bd{z}}\bd{Q}_s$ is symmetric using the following argument. Eq.~\eqref{eq:symplectic} implies, with the subscript $\ts$ omitted for convenience,
\begin{align}
& \partial_{\bd{q}}\bd{Q}(\partial_{\bd{q}}\bd{P})^\TT - \partial_{\bd{q}}\bd{P}(\partial_{\bd{q}}\bd{Q})^\TT = 0_{3\times 3}, \label{eq:p1}\\
& \partial_{\bd{q}}\bd{Q}(\partial_{\bd{p}}\bd{P})^\TT - \partial_{\bd{q}}\bd{P}(\partial_{\bd{p}}\bd{Q})^\TT = \Id_3, \label{eq:p2} \\
& \partial_{\bd{p}}\bd{Q}(\partial_{\bd{q}}\bd{P})^\TT - \partial_{\bd{p}}\bd{P}(\partial_{\bd{q}}\bd{Q})^\TT = - \Id_3, \label{eq:p3} \\
& \partial_{\bd{p}}\bd{Q}(\partial_{\bd{p}}\bd{P})^\TT - \partial_{\bd{p}}\bd{P}(\partial_{\bd{p}}\bd{Q})^\TT = 0_{3\times 3}, \label{eq:p4}
\end{align}
where $0_{3\times 3}$ is $3$-by-$3$ zero matrix.

\noindent Eq.~\eqref{eq:p1}$-\I\times$ Eq.~\eqref{eq:p3} gives
\begin{equation}\label{eq:p5}
\partial_{\bd{z}}\bd{Q}(\partial_{\bd{q}}\bd{P})^\TT - \partial_{\bd{z}}\bd{P}(\partial_{\bd{q}}\bd{Q})^\TT = \I\Id_3.
\end{equation}

\noindent Eq.~\eqref{eq:p2}$-\I\times$ Eq.~\eqref{eq:p4} gives
\begin{equation}\label{eq:p6}
\partial_{\bd{z}}\bd{Q}(\partial_{\bd{p}}\bd{P})^\TT - \partial_{\bd{z}}\bd{P}(\partial_{\bd{p}}\bd{Q})^\TT = \Id_3.
\end{equation}

\noindent Eq.~\eqref{eq:p5}$-\I\times$ Eq.~\eqref{eq:p6} gives
\begin{equation*}
\partial_{\bd{z}}\bd{Q}(\partial_{\bd{z}}\bd{P})^\TT - \partial_{\bd{z}}\bd{P}(\partial_{\bd{z}}\bd{Q})^\TT = 0_{3\times 3}.
\end{equation*}
Combined with $\partial_{\bd{z}}\bd{Q}(\partial_{\bd{z}}\bd{Q})^\TT - \partial_{\bd{z}}\bd{Q}(\partial_{\bd{z}}\bd{Q})^\TT = 0_{3\times 3}$, one has
\begin{equation*}
\partial_{\bd{z}}QZ^\TT-Z(\partial_{\bd{z}}Q)^\TT=0,
\end{equation*}
which implies $Z^{-1}\partial_{\bd{z}}Q=(\partial_{\bd{z}}Q)^\TT(Z^\TT)^{-1}=(\partial_{\bd{z}}Q)^\TT(Z^{-1})^\TT=(Z^{-1}\partial_{\bd{z}}Q)^\TT$. Therefore, $Z^{-1}\partial_{\bd{z}}Q$ is symmetric, and then $m_{\tsh\rightarrow\tsv}=m_{\tsv\rightarrow\tsh}=0$, which brings eqs.~\eqref{eq:amp_SV} and \eqref{eq:amp_SH} by eq.~\eqref{eq:2}.

\subsection{Auxiliary Operators}\label{4.2}
In this section the asymptotic is expanded and done without the dynamics.  The necessary operators are derived and proof of convergence is shown. Starting from equation~\eqref{FGA:1}
\begin{multline}\label{weq:2.1}
\rho \Big(\bd{A}_{tt} +  \ds\frac{2}{\epsilon}\I\bd{A}_t\left((\bP_t- \I \bQ_t)\cdot (\bd{x}-\bQ) - \bP\cdot \bQ_t \right) \\
+  \ds\frac{\I}{\epsilon} \bd{A} \left( (\bP_{tt}- \I \bQ_{tt})\cdot (\bd{x}-\bQ) - (\bP_t- \I \bQ_t)\cdot \bQ_t\right) \\
-  \ds\frac{1}{\epsilon^2}\bd{A} \left(\big[(\bP_t- \I \bQ_t)\cdot (\bd{x}-\bQ)\big]^2 + (\bP\cdot \bQ_t)^2 - 2(\bP\cdot \bQ_t)\big[(\bP_t- \I \bQ_t)\cdot (\bd{x}-\bQ)\big] \right) \Big) \\
\sim   (\lambda + 2\mu)\left(\ds\frac{\I}{\epsilon}\nabla(\bd{A}\cdot (\bP + \I (\bd{x} - \bQ))) - \ds\frac{1}{\epsilon^2}(\bd{A} \cdot(\bP + \I (\bd{x} - \bQ)))(\bP + \I (\bd{x} - \bQ)) \right)  \\ 
-\mu\left(\ds\frac{\I}{\epsilon}  \Curl((\bP + \I (\bd{x} - \bQ)) \times \bd{A}) - \ds\frac{1}{\epsilon^2} \nabla (\bP + \I (\bd{x} - \bQ))\times((\bP + \I (\bd{x} - \bQ))\times \bd{A})\right). 
\end{multline}
\noindent Expanding $\rho(\bd{x})$ around $\bQ$ and truncating at order third order, Grouping in terms of $(\bx - \bQ)$, up to $\cO(3)$, and dropping terms that produce terms higher than $\cO(\epsilon)$, we can rewrite with simplifying, starting with $\cO((\bx-\bQ)^0)$
\begin{equation}\label{weq:6:1}
\rho M_0+ \ds\frac{\lambda}{\epsilon}\bd{A} + (\lambda + \mu)\ds\frac{1}{\epsilon^2}(\bd{A} \cdot\bP)\bP + \ds\frac{\mu}{\epsilon^2} (\bP\cdot\bP)\bd{A} \\
\end{equation}
The $\cO((\bx-\bQ)^1)$ term is
\begin{multline}\label{weq:6:2}%
\rho M_1\cdot (\bd{x}-\bQ) + (\lambda + \mu)\ds\frac{\I}{\epsilon^2} (\bd{A}\cdot(\bd{x} - \bQ))\bP + (\lambda + \mu)\ds\frac{\I}{\epsilon^2}(\bd{A} \cdot\bP)(\bd{x} - \bQ) \\
+ \ds\frac{2\mu\I}{\epsilon^2}(\bP\cdot(\bd{x} - \bQ))\bd{A}+ \partial_{\bQ}\rho \cdot(\bd{x}-\bQ)M_0
\end{multline}
The $\cO((\bx-\bQ)^2)$ term,
\begin{multline}\label{weq:6:3}%
\rho\left(-\ds\frac{1}{\epsilon^2}\bd{A}\big[(\bP_t- \I \bQ_t)\cdot (\bd{x}-\bQ)\big]^2\right) \\
- (\lambda + \mu)\ds\frac{1}{\epsilon^2}(\bd{A}\cdot(\bd{x} - \bQ)) (\bd{x} - \bQ) - \frac{\mu}{\epsilon^2}|\bd{x} - \bQ|^2\bd{A} \\
+ \partial_{\bQ}\rho \cdot(\bd{x}-\bQ)M_1\cdot (\bd{x}-\bQ) + \ds\frac{1}{2}(\bd{x}-\bQ)\cdot \partial_{\bQ\bQ}^2\rho(\bd{x}-\bQ)M_0 
\end{multline}
Finally the $\cO((\bx-\bQ)^3)$ term,
\begin{equation}\label{weq:6:4}%
(\bd{x}-\bQ)\cdot \partial_{\bQ\bQ}^2\rho(\bd{x}-\bQ)\left(\ds\frac{\bd{A}}{\epsilon^2}(\bP\cdot \bQ_t)\big[(\bP_t- \I \bQ_t)\cdot (\bd{x}-\bQ)\big] 
\right)
\end{equation}
With $M_0$ and $M_1$ defined as
\begin{multline}
M_0 = \bd{A}_{tt} - \bd{A}_{t}\frac{2\I}{\epsilon}\bP\cdot \bQ_t - \bd{A}\frac{\I}{\epsilon}(\bP_t- \I \bQ_t)\cdot \bQ_t -  \ds\frac{1}{\epsilon^2}\bd{A}(\bP\cdot \bQ_t)^2 \\
M_1 = \ds\frac{2}{\epsilon}\I\bd{A}_t(\bP_t- \I \bQ_t) + \ds\frac{\I}{\epsilon}\bd{A}  (\bP_{tt}- \I \bQ_{tt})+ \frac{2}{\epsilon}\bd{A}(\bP\cdot \bQ_t)\big[(\bP_t- \I \bQ_t)\big]
\end{multline}
Now applying lemma~\ref{lemma:2}, first applying integration by parts to $M_1\cdot (\bx-\bQ)$ to clarify notation
\begin{equation}
M_1\cdot (\bx-\bQ) = -\epsilon \partial_l (M_{1,ij} Z_{jl}^{-1}) = -\epsilon \partial_\bz:(Z^{-\TT} M_1) 
\end{equation}
where $:$ denotes a contraction of the indicates. Considering each order separately, first the $\cO((\bx-\bQ)^1)$
\begin{multline}
-\epsilon \rho\partial_\bz :(\rho Z^{-\TT} M_1) - (\lambda + \mu)\ds\frac{\I}{\epsilon}\partial_\bz:(Z^{-\TT}\bd{A}\bP) - (\lambda + \mu)\ds\frac{\I}{\epsilon}\partial_\bz:(\bd{A} \cdot\bP Z^{-\TT}) \\
- \ds\frac{2\mu\I}{\epsilon}\partial_\bz:(Z^{-\TT}\bP\bd{A})- \epsilon\partial_\bz:(Z^{-\TT}\partial_{\bQ}\rho M_0) 
\end{multline}
The $\cO((\bx-\bQ)^2)$ term
\begin{multline}
\ds\frac{-\rho}{\epsilon}\bd{A}\tr\left((\bP_t- \I \bQ_t)\otimes(\bP_t- \I \bQ_t)Z^{-1}\partial_{\bz} \bQ\right) \\
+ \partial_{\bz}:(\partial_{\bz}:(\rho\bd{A}(\bP_t- \I \bQ_t)\otimes(\bP_t- \I \bQ_t)Z^{-1})Z^{-1})\\
- (\lambda + \mu)\ds\frac{1}{\epsilon}Z^{-1}:(\partial_{\bz}\bQ \bd{A})  - (\lambda + \mu)\partial_{\bz}:(\partial_{\bz}:(\partial_{\bz}(\bQ A)Z^{-1})Z^{-1}) \\
-\frac{\mu}{\epsilon}\bd{A}\tr(Z^{-1}\partial_{\bz}) -\mu\partial_{\bz}:(\partial_{\bz}:(Z^{-1})\bd{A}Z^{-1}) \\
+ \epsilon\partial_{\bQ}\rho M_1\tr(Z^{-1}\partial_{\bz})  + \epsilon\ds\frac{1}{2}\partial_{\bQ\bQ}^2\rho M_0\tr(Z^{-1}\partial_{\bz}) \\
+ \epsilon^2\partial_{\bz}:(\partial_{\bz}:(Z^{-1}\partial_{\bQ\bQ}^2\rho)M_0Z^{-1}) 
\end{multline}
The $\cO((\bx-\bQ)^3)$ term only up to first order terms
\begin{equation}
\partial_{\bz}:\left(\tr(\partial_{\bQ\bQ}^2\rho Z^{-t}\partial_{\bz}\bQ)(\bP\cdot \bQ_t)\big[(\bP_t- \I \bQ_t)\big] \bd{A}\right)
\end{equation}
Expanding $M_0$ and $M_1$ when needed and grouping in terms of $\epsilon$
\begin{align*}
\ds\frac{1}{\epsilon^2}& \Big[-\rho\bd{A}(\bP\cdot \bQ_t)^2 + (\lambda + \mu)(\bd{A} \cdot\bP)\bP + \mu (\bP\cdot\bP)\bd{A}\Big] \\
+ \ds\frac{1}{\epsilon}&  \Big[\lambda\bd{A}+ 2\I\bd{A}_{t}\bP\cdot \bQ_t -\I\rho\bd{A}(\bP_t- \I \bQ_t)\cdot \bQ_t \\
&- \I(\lambda + \mu)\partial_\bz:(Z^{-\TT}\bd{A}\bP)  - \I(\lambda + \mu)\partial_\bz:(\bd{A} \cdot\bP Z^{-\TT}) \\
&+ \partial_\bz:(Z^{-\TT}\partial_{\bQ}\rho \bd{A}(\bP\cdot \bQ_t)^2) - 2\mu\I\partial_\bz:(Z^{-\TT}\bP\bd{A}) \\
&-\rho\bd{A}\tr\left((\bP_t- \I \bQ_t)\otimes(\bP_t- \I \bQ_t)Z^{-1}\partial_{\bz} \bQ\right) \\
&- (\lambda + \mu)Z^{-1}:(\partial_{\bz}\bQ \bd{A}) -\mu\bd{A}\tr(Z^{-1}\partial_{\bz}) \\
&-\ds\frac{1}{2}\partial_{\bQ\bQ}^2\rho \bd{A}(\bP\cdot \bQ_t)^2\tr(Z^{-1}\partial_{\bz}) \Big] \\
+\epsilon^0& \Big[\rho\bd{A}_{tt}- \rho\bd{A}_{t}\bP\cdot \bQ_t  -\epsilon \rho\partial_\bz :(\rho Z^{-\TT} M_1)  \\
&+ \partial_\bz:(Z^{-\TT}\partial_{\bQ}\rho(\bP_t- \I \bQ_t)\cdot \bQ_t)\bd{A}) \\
&+ \partial_{\bz}:(\partial_{\bz}:(\rho\bd{A}(\bP_t- \I \bQ_t)\otimes(\bP_t- \I \bQ_t)Z^{-1})Z^{-1}) \\
&- (\lambda + \mu)\partial_{\bz}:(\partial_{\bz}:(\partial_{\bz}(\bQ A)Z^{-1})Z^{-1}) \\
&-\mu\partial_{\bz}:(\partial_{\bz}:(Z^{-1})\bd{A}Z^{-1}) + \epsilon\partial_{\bQ}\rho M_1\tr(Z^{-1}\partial_{\bz}) \\
&- \I\ds\frac{1}{2}\partial_{\bQ\bQ}^2\rho\bd{A}(\bP_t- \I \bQ_t)\cdot \bQ_t\tr(Z^{-1}\partial_{\bz}) \\
&-\partial_{\bz}:(\partial_{\bz}:(Z^{-1}\partial_{\bQ\bQ}^2\rho)\bd{A}(\bP\cdot \bQ_t)^2Z^{-1}) \\
&+\partial_{\bz}:\left(\tr(\partial_{\bQ\bQ}^2\rho Z^{-t}\partial_{\bz}\bQ)(\bP\cdot \bQ_t)\big[(\bP_t- \I \bQ_t)\big] \bd{A}\right) \Big]
\end{align*}
Now define the operators $\cL_0, \cL_1, \cL_2$ acting on $\bd{A}$
\begin{equation}
\cL_0(\bd{A}):=-\rho\bd{A}(\bP\cdot \bQ_t)^2 + (\lambda + \mu)(\bd{A} \cdot\bP)\bP + \mu (\bP\cdot\bP)\bd{A}
\end{equation}
\begin{multline}
\cL_1(\bd{A}) := \lambda\bd{A}+ 2\I\bd{A}_{t}\bP\cdot \bQ_t -\I\rho\bd{A}(\bP_t- \I \bQ_t)\cdot \bQ_t \\
 - \I(\lambda + \mu)\partial_\bz:(Z^{-\TT}\bd{A}\bP)  - \I(\lambda + \mu)\partial_\bz:(\bd{A} \cdot\bP Z^{-\TT}) \\
+ \partial_\bz:(Z^{-\TT}\partial_{\bQ}\rho \bd{A}(\bP\cdot \bQ_t)^2) - 2\mu\I\partial_\bz:(Z^{-\TT}\bP\bd{A}) \\
-\rho\bd{A}\tr\left((\bP_t- \I \bQ_t)\otimes(\bP_t- \I \bQ_t)Z^{-1}\partial_{\bz} \bQ\right)\\
 - (\lambda + \mu)Z^{-1}:(\partial_{\bz}\bQ \bd{A}) -\mu\bd{A}\tr(Z^{-1}\partial_{\bz}) \\
-\ds\frac{1}{2}\partial_{\bQ\bQ}^2\rho \bd{A}(\bP\cdot \bQ_t)^2\tr(Z^{-1}\partial_{\bz})
\end{multline}
\begin{multline}
\cL_2(\bd{A}) := \rho\bd{A}_{tt}- \rho\bd{A}_{t}\bP\cdot \bQ_t  -\epsilon \rho\partial_\bz :(\rho Z^{-\TT} M_1) \\
 + \partial_\bz:(Z^{-\TT}\partial_{\bQ}\rho(\bP_t- \I \bQ_t)\cdot \bQ_t)\bd{A}) \\
+ \partial_{\bz}:(\partial_{\bz}:(\rho\bd{A}(\bP_t- \I \bQ_t)\otimes(\bP_t- \I \bQ_t)Z^{-1})Z^{-1})) \\
- (\lambda + \mu)\partial_{\bz}:(\partial_{\bz}:(\partial_{\bz}(\bQ A)Z^{-1})Z^{-1}) \\
-\mu\partial_{\bz}:(\partial_{\bz}:(Z^{-1})\bd{A}Z^{-1}) + \epsilon\partial_{\bQ}\rho M_1\tr(Z^{-1}\partial_{\bz}) \\
- \I\ds\frac{1}{2}\partial_{\bQ\bQ}^2\rho\bd{A}(\bP_t- \I \bQ_t)\cdot \bQ_t\tr(Z^{-1}\partial_{\bz}) \\
-\partial_{\bz}:(\partial_{\bz}:(Z^{-1}\partial_{\bQ\bQ}^2\rho)\bd{A}(\bP\cdot \bQ_t)^2Z^{-1}) \\
+\partial_{\bz}:\left(\tr(\partial_{\bQ\bQ}^2\rho Z^{-t}\partial_{\bz}\bQ)(\bP\cdot \bQ_t)\big[(\bP_t- \I \bQ_t)\big] \bd{A}\right) 
\end{multline}
Now $(\partial_t^2 - \cL) \bu_{\FGA}$ can be written as
\begin{multline}
(\partial_t^2 - \cL)\bu_{\FGA} = (2\pi\epsilon)^{-3/2}\sum_{n}\int_{\mR^{3d}} (\epsilon^{-2}\cL_{n,0}(\ba_{n,0} + \epsilon\ba_{n,1}) \\
\epsilon^{-1}\cL_{n,1}(\ba_{n,0} + \epsilon\ba_{n,1}) + \cL_{n,2}(\ba_{n,0} + \epsilon\ba_{n,1}))G_n^\epsilon\ud \by\ud\bq\ud\bp
\end{multline}
Substituting the dynamics for $\cL_{n,0}$ reveals that that $\cL_{n,0}(\ba_{n,0}) = 0$. Looking at the $\cO(1/\epsilon)$ term and equating to zero gives
\begin{equation}\label{eq:order-1}
\cL_{n,1}(\ba_{n,0}) = -\cL_{n,0}(\ba_{n,1}) 
\end{equation}
Now $\cL_{n,0}$ is defined as
\begin{equation}
\cL_{n,0} = \big(\mu |\bP_n|^2-\rho(\bP_n\cdot\partial_t\bQ_n)^2\big)\Id_3 + (\lambda + \mu)\bP_n\otimes \bP_n
\end{equation}
which is a symmetric matrix with eigenvalues
\begin{align}
\beta_{n,1} &= (\lambda + 2\mu)|\bP_n|^2 -\rho|\bP\cdot\partial_t\bQ_n|^2 \\
 \beta_{n,2} &= \mu|\bP|^2  -\rho|\bP_n\cdot\partial_t\bQ_n|^2  \\
\beta_{n,3} &=\mu|\bP|^2 -\rho|\bP_n\cdot\partial_t\bQ_n|^2
\end{align}
the corresponding eigenvectors are,
\begin{align}
\bP_n &= (p_{n,1},p_{n,1},p_{n,1})\\
\bd{d}_{n,1} &= (-p_{n,2},p_{n,1},0) \\
\bd{d}_{n,2} &= (-p_{n,3},0,p_{n,1})  .
\end{align}
For the P-wave, $n = \tp$, taking inner product of with the eigenvectors;
\begin{equation}
\ip{\bP_\tp}{\cL_{\tp,0}(\ba_{n,1})}=-\ip{\bP_\tp}{\cL_{\tp,1}(\ba_{\tp,0})} 
\end{equation}
This gives,
\begin{multline}
\ip{\cL_{\tp,0}^*(\bP_\tp)}{\ba_{n,1}}=\ip{\cL_{\tp,0}(\bP)}{\ba_{\tp,1}} \\
= ((\lambda + 2\mu)|\bP_\tp|^2 -\rho|\bP_\tp\cdot\partial_t\bQ_\tp|^2)\ip{\bP_\tp}{\ba_{\tp,1}} = 0
\end{multline}
after plugging in the dynamics so we can recover the equation~\eqref{eq:amp_P}, as 
\begin{equation}
\ip{\bP_\tp}{\cL_{\tp,1}(\ba_{\tp,0})} = 0.
\end{equation}
Considering $\bd{d}_{1,2}$,
\begin{equation}
\ip{\bd{d}_{1,2}}{\cL_{\tp,0}(\ba_{n,1})}=-\ip{\bd{d}_{1,2}}{\cL_{\tp,1}(\ba_{\tp,0})}.
\end{equation}
Then
\begin{equation}
\ip{\cL_{\tp,0}^*(\bd{d}_{1,2})}{\ba_{\tp,1}}=\ip{\cL_{\tp,0}(\bd{d}_{1,2})}{\ba_{\tp,1}} =  (\mu|\bP|^2  -\rho|\bP\cdot\partial_t\bQ|^2)\ip{\bd{d}_{1,2}}{\ba_{\tp,1}}
\end{equation}
And so plugging in the flow
\begin{equation}
\ip{\bd{d}_{1,2}}{\ba_{\tp,1}} = \ds\frac{1}{\rho(c_\ts^2 - c_\tp^2)|\bP|^2}\ip{\bd{d}_{1,2}}{\cL_{\tp,1}(\ba_{\tp,0})}
\end{equation}
Define the pseudo-inverse, for $\bv\in \cS(\mR^3)$
\begin{equation}
\cL^{-1}_{\tp,0}(\bv) = \ds\frac{1}{\rho(c_\ts^2 - c_\tp^2)|\bP|^2}\left(\ip{\hat{\bd{d}}_1}{\bv}\hat{\bd{d}}_1 + \ip{\hat{\bd{d}}_2}{\bv}\hat{\bd{d}}_2\right)
\end{equation}
and define
\begin{equation}
\ba_{\tp}^\perp \bv = \cL^{-1}_{\tp,0}\big((\Id - \Pi_\tp) \cL_{\tp,1}(\bv) \big)
\end{equation}
Where $\Pi_\tp$ is projection onto $\bP_\tp$. \\
For the S-wave, $n = \tsv,\tsh$. From~\eqref{eq:order-1} we have
\begin{equation}\label{eq:order-1s}
\cL_{\ts,1}(\ba_{\ts}) = -\cL_{\ts,0}(\ba_{\tsh,0} + \ba_{\tsv,0}) 
\end{equation}
Let $\bd{d}_{s,1} = \hat{\bd{N}}_\tsh$, taking inner product with~\eqref{eq:order-1s} gives
\begin{equation}
\ip{\cL_{\ts,0}(\hat{\bd{N}}_{\tsv})}{\ba_{\ts,1}} = (\mu|\bP_{\ts}|^2 -\rho|\bP_{\ts}\cdot\partial_t\bQ_{\ts}|^2)\ip{\hat{\bd{N}}_{\tsv}}{\ba_{\ts,1}} = 0,
\end{equation}
which is zero when the dynamics are substituted.  From this we can get 
\begin{equation}
\ip{\hat{\bd{N}}_{\tsv}}{\cL_{\ts,1}\ba_{\tsv,0}} =  -\ip{\hat{\bd{N}}_{\tsv}}{\cL_{\ts,1}\ba_{\tsh,0}}
\end{equation}
Which gives us equation~\ref{eq:amp_SV}, equation ~\ref{eq:amp_SH} can be recovered in a similar manor. Taking inner product with $\bP_\ts$ of~\eqref{eq:order-1s} leads to,
\begin{equation}
\ip{\bP}{\ba_{\ts,1}} = -\ds\frac{1}{(\lambda + \mu)|\bP|^2}\ip{\bP}{\cL_{\ts,1}(\ba_{\tsv,0} + \ba_{\tsh,0})}
\end{equation}
Define the pseudo-inverse for $\bv\in\cS(\mR^3)$
\begin{equation}
\cL^{-1}_{\ts,0}(\bv) = -\ds\frac{1}{(\lambda + \mu)|\bP_\ts|^2}\ip{\hat{\bP_{\ts}}}{\bv}\hat{\bP_{\ts}}
\end{equation}
and define
\begin{equation}
\ba_{\ts}^\perp \bv = \cL^{-1}_{\ts,0}\big((\Id - \Pi_\ts) \cL_{\ts,1}(\bv) \big).
\end{equation}
with $\Pi_\ts$ a projection onto the span of $\bd{d}_{\ts,1}$ and $\bd{d}_{\ts,2}$.
\section{Error Estimates and Main Result}\label{sec:5}
\begin{definition}Define the scaled semi-norm;
\begin{equation}\label{eq:enorm} 
\|\bu(t,\cdot)\|_{\rE} = \epsilon(\|\partial_t\bu(t,\cdot)\|_{\rL^2} + \|\Div\bu(t,\cdot)\|_{\rL^2} + \|\Curl\bu(t,\cdot)\|_{\rL^2})
\end{equation}
\end{definition}
\begin{proposition}\label{eq:a_bd}
 Let $\ba_\ts = a_\tsv\alpha_{\tsv}\hat{\bd{N}}_{\tsv} + a_\tsh\alpha_{\tsh}\hat{\bd{N}}_{\tsh}$ and $\ba_\tp = a_\tp\alpha_{\tp}{\hat{\bd{N}}}_{\tp}$. The terms $\ba_\tp$, $\ba_\ts$ are bounded in the $\rL^2$ sense; furthermore,
\begin{equation}
\| u_{\FGA,1}- u_{\FGA,0}\|_{\rE} \leq \epsilon C_{T,\delta}
\end{equation}
\end{proposition}
\begin{proof} 
First we remark
\begin{equation}
\|\ba_n(t,\cdot)\|_{\rL^2} \leq \|\alpha_{n}\|_{\rL^2}\|a_n(t,\cdot)\|_{\rL^\infty} \text{ and } \|\ba_n(t,\cdot)\|_{\rL^\infty} \lesssim \|a_n(t,\cdot)\|_{\rL^\infty}
\end{equation}
From the definitions we have an immediate bound
\begin{align}\label{fga:P_0-P_1}
\|u_{\FGA,1}(t,\cdot) - u_{\FGA,0}(t,\cdot)\|_{\rE} \leq  (2\pi\epsilon)^{-3d/2}\sum_{n}\epsilon\|\int_{\mR^{3d}}\ba_{n,1}^\perp + \ba_{n,1}G_n\ud \by\ud\bq\ud\bp\|_{\rE}
\end{align}
Applying the derivatives with proposition~\ref{eq:fio_bd}, we have the bound
\begin{equation}
\|u_{\FGA,1}(t,\cdot) - u_{\FGA,0}(t,\cdot)\|_{\rE} \leq \epsilon C\sum_{n}\|\ba_{n,1}^\perp(t,\cdot) + \ba_{n,1}(t,\cdot)\|_{\rL^\infty}
\end{equation}
The estimate of~\eqref{fga:P_0-P_1} then follows directly from Proposition 3.7 in~\cite{LuYa:CPAM}. We need to bound the prefactor terms, we note that on the compact set $K_{\delta}$ the bound for the prefactor terms fall from Lemma 5.4 in~\cite{LuYa:CPAM}, We go through several of the bounds here, starting with the P-wave and dropping the subscripts as the calculations are the same and setting $\bP = \bP_p$, $\bQ = \bQ_p$,
\begin{equation}\label{eq:a0_1}
\partial_t a_{\tp,0} = a_{\tp,0}\left(\ds\frac{\partial_{\bQ} c_{\tp}\cdot \bP}{|\bP|} + \ds\frac{1}{2}\tr\left(Z^{-1}\partial_t Z\right)\right)
\end{equation}
\begin{equation}\label{eq:a1_1}
\partial_t a_{\tp,1} = a_{\tp,1}\left(\ds\frac{\partial_{\bQ} c_{\tp}\cdot \bP}{|\bP|} + \ds\frac{1}{2}\tr\left(Z^{-1}\partial_t Z\right)\right) + F_p(a_{\tp,0}, \partial_{\bz} a_{\tp,0},\bQ,\bP,c_\tp)
\end{equation}
With $F$ being a continuous function in its arguments for $P,Q\in K_{\delta_T}$. Equation~\eqref{eq:a0_1} immediately implies;
\begin{equation}\label{eq:a0_2}
\partial_t |a_{\tp,0}| \leq |a_{\tp,0}|\left|\ds\frac{\partial_{\bQ} c_{\tp}\cdot \bP}{|\bP|} + \ds\frac{1}{2}\tr\left(Z^{-1}\partial_t Z\right)\right|
\end{equation}
an application of Gronwalls gives with a computation of $\partial_{\bz}$, 
\begin{equation}
\sup_{t\in[0,T]} \Lambda_{1,K_{\delta/2}}\big( a_{\tp,0}(t,\bq,\bp)) \leq  C_{\delta,T}
\end{equation}
To bound Equation~\eqref{eq:a1_1} $\partial_{\bz} a_{\tp,0}$ needs to be bounded, but Equation~\eqref{eq:a0_2} shows.
\begin{equation}
\partial_t |\partial_{\bz} a_{\tp,0}| \leq |\partial_{\bz} a_{\tp,0}|\left|\partial_{\bz}\left(\ds\frac{\partial_{\bQ} c_{\tp}\cdot \bP}{|\bP|} + \ds\frac{1}{2}\tr\left(Z^{-1}\partial_t Z\right)\right)\right|
\end{equation}
Then Gronwalls gives, 
\begin{equation}
\sup_{t\in[0,T]} \Lambda_{1,K_{\delta/2}}\big( a_{\tp,0}(t,\bq,\bp)\big) \leq  C_{\delta,T}
\end{equation}
The function $F(a_{\tp,0}, \partial_{\bz} a_{\tp,0},\bQ,\bP,c_\tp)$ is differentiable with differentiable arguments on the compact set $K_{\delta/2}$. Combining these we do have
\begin{equation}
\partial_t|a_{\tp,1}| \leq C_{T,\delta}
\end{equation}
Combining these we have that $\|a_\tp\bP\|_{\rL^2}$ is bounded on $[0,T]$, $K_{\delta/2}$. For the S-wave terms, again using the short notation $\bP = \bP_s$, $\bQ = \bQ_s$ and dropping the brach subscript, we can write the system (\eqref{eq:amp_SV},~\eqref{eq:amp_SH}) as
\begin{align}\label{fga_amp:sys}
\ddt{}\left(\begin{array}{c}
\ba^{sv}_{\pm} \\
\ba^{sh}_{\pm} 
\end{array}\right) &= 
\frac{1}{2}\left(\begin{array}{c c}
h_{\pm} & m_{\pm} \\
-m_{\pm} & h_{\pm} \\
\end{array}\right)
\left(\begin{array}{c}
\ba^{sv}_{\pm} \\
\ba^{sh}_{\pm} 
\end{array}\right)
\end{align}
where $m_{\pm} = \partial_t\hat{\bd{N}}_{\tsh}\cdot \hat{\bd{N}}_{\tsv}$ and
\begin{align}
h_{\pm} = 2\partial_{\bQ_{\ts,\pm}}c_s\cdot \hat{\bd{N}}_{\tsh} + \ba^s\tr\left(\cZ^{-1}\partial_t \cZ^s\right)
\end{align}
Denote $M$ as the matrix in eq.~\eqref{fga_amp:sys}, and $\ba = (\ba^{sv},\ba^{sh})^T$.  Then the system can be recast as
\begin{equation}
\ddt{\ba} = M(t)\ba
\end{equation}
Solving for the eigenvalues:
\begin{equation}\label{eq:eigs_s}
\lambda_{\tsh,\tsv}(t) = -\partial_{\bd{Q}}c_{\ts}\cdot\hat{\bd{N}}_{\tsh,\tsv} - \frac{1}{2}\tr\Bigl(Z_{\ts}^{-1}\frac{\ud Z_{\ts}}{\ud t}\Bigr) \mp \I\frac{\ud \hat{\bd{N}}_{\tsv,\tsv}}{\ud t}\cdot\hat{\bd{N}}_{\tsh,\tsv}
\end{equation} 
To see that these are bounded simply note that form a smooth $\{\hat{N}_{\tp}, \hat{N}_{\tsh}, \hat{N}_{\tsh}\}$ form an orthonormal frame and hence the last term in~\eqref{eq:eigs_s} is bounded for all $t \geq 0$.
\begin{equation}\label{eq:Zs}
\tr\Bigl(Z_{\ts}^{-1}\frac{\ud Z_{\ts}}{\ud t}\Bigr) = \ds\frac{1}{\det(Z_{\ts})}\ddt{\det(Z_{\ts})},
\end{equation}
then by~\eqref{lemma:1} we have a bound for $\det(Z_{\ts})$ so Eq.~\eqref{eq:Zs} is bounded for all $t \geq 0$ and 
\begin{equation}\label{eq:hm_bd}
\ds\frac{\partial_{\bq} H_{\ts}\cdot \partial_{\bp} H_{\ts}}{H_{\ts}} = -\frac{\partial_{\bd{Q}_{\ts}}c_{\ts}\cdot \bd{P}_{\ts}}{|\bd{P}_{\ts}|} 
\end{equation}
by observation, with ~\eqref{eq:hambound}, Eq.~\eqref{eq:hm_bd} is bounded for all $t \geq 0$. The the eigenvalues in Eq.~\eqref{eq:eigs_s} are bounded for all $t \geq 0$. So we have the
\begin{equation}
\sup_{t\in[0,T]} \Lambda_{0,K_{\delta/2}}\big(a_{\ts,0}(t,\bq,\bp)\big) \leq  C_{\delta,T}
\end{equation}
\noindent For $\partial_t\ba_{\ts,1}$, we have the system
\begin{align}\label{fga_amp:sys1}
\ddt{}\left(\begin{array}{c}
a_{\tsv,1} \\
a_{\tsh,1} 
\end{array}\right) &= 
\frac{1}{2}\left(\begin{array}{c c}
h & m \\
-m & h \\
\end{array}\right)
\left(\begin{array}{c}
a^{\tsv}_{\pm} \\
a^{\tsh}_{\pm} 
\end{array}\right) + 
\bF(\ba_{\ts,0},\partial_{\bz} \ba_{\ts,0},\bQ,\bP,c_\ts)
\end{align}
Then the bounds follow from previous work, as $\bF$ is smooth away from $|\bP|>0$, so we arrive at 
\begin{equation}
\sup_{t\in[0,T]} \Lambda_{1,K_{\delta/2}}\big(a_{\ts,1}(t,\bq,\bp)\big) \leq  C_{\delta,T}
\end{equation}
\end{proof}
\begin{proposition}\label{prop:e_inf}
For any $T > 0$ and $t\in [0,T]$
\begin{equation}
\| \bu_{\FGA,1}(t,\cdot) - \tilde{\bu}_{\FGA,1}(t,\cdot) \|_{\rE} = \cO(\epsilon^\infty)
\end{equation}
\end{proposition}
\begin{proof}
Starting from the definition,
\begin{align}
&  \|\bu_{\FGA,1} - \tilde{\bu}_{\FGA,1} \|_{\rE} \\
&\leq \ (2\pi\epsilon)^{-3d/2}\int_{\mR^{3d}} \Big\|\sum_{n} (1 - \chi_\delta)(\ba_{n,0} + \epsilon \ba_{n,1}) e^{\frac{\I}{\epsilon} \Phi_{n}}  \Big\|_{\rE} \ud\by\ud\bp\ud\bq \nonumber \\
&\leq 2^{-d/2}\sum_{n}\|(1 - \chi_\delta)(a_{n,0} + \epsilon a_{n,1})\hat{\bd{N}}_{n}\cF\alpha_{n} \|_{\rE}  \nonumber \\
&\leq 2^{-d/2}\Big(\|\cF^\epsilon u_0^\epsilon \|_{\rL^2(\mR^n\bs K_\delta)} + \epsilon\|\cF^\epsilon u_1^\epsilon \|_{\rL^2(\mR^n\bs K_\delta)}\Big)\sum_n\|(1 - \chi_\delta)(a_{n,0}+ \epsilon a_{n,1})\hat{\bd{N}}_{n}\|_{\rL^\infty}  \nonumber \\
& \leq C_{\delta,T}\Big(\|\cF^\epsilon u_0^\epsilon \|_{\rE(\mR^n\bs K_\delta)} + \epsilon\|\cF^\epsilon u_1^\epsilon \|_{\rE(\mR^n\bs K_\delta)}\Big) \\
& \leq C_{\delta,T}\Big(\|\cF^\epsilon u_0^\epsilon \|_{\rL^2(\mR^n\bs K_\delta)} + \epsilon\|\cF^\epsilon u_1^\epsilon \|_{\rL^2(\mR^n\bs K_\delta)}\Big) = \cO(\epsilon^\infty)
\end{align}
Where the second equality is from Proposition~\eqref{eq:fio_bd}, the third inequality is by similar arguments found in Proposition~\eqref{eq:a_bd}. Also from Eq.~\eqref{eq:alpha}, with the direct bound
\begin{equation}\label{eq:alpha_est}
\|\alpha_n\|_{\rE} \leq C(\|\bu_0^\epsilon\|_{\rE} + \epsilon \|\bu_1^\epsilon\|_{\rE}).
\end{equation}
The last inequality is justified noticing the derivatives do not affect the initial conditions only $G_n^\epsilon$ in  $\cF^\epsilon u_0^\epsilon$  and  $\cF^\epsilon u_1^\epsilon$. 
\end{proof}
\begin{proposition}\label{prop:L_bd} The operators $\cL_0, \cL_1, \cL_2$ are bounded. That is, for a given $T$ and any $t \in [0,T]$, $\ba \in C^\infty([0,T])\times S(\mR^{2d})$ and for $k = 0, 1$, $j = 0, 1, 2$
\begin{equation}
\sup_{t\in[0,T]}\Lambda_{k,K_{\delta}}\big(\cL_j(\ba)) < C_{T,K_{\delta}} \text{ and } \|\cL_j(\ba(t,\cdot))\|_{\rL^\infty} < C_{T,\delta}
\end{equation}
\end{proposition}
\begin{proof}
Notice $\cL_{n,j}$ depend on $\bP_n$, $\bQ_n$ $Z_n^{-1}$ it its derivatives, all of which are bounded on $[0,T]\times K_{\delta}$, this gives the result
\end{proof}
\begin{proposition}\label{prop:a_perp} For a given $T$ and any $t \in [0,T]$, for $k = 0, 1$ and $\ba \in C^\infty([0,T])\times S(\mR^{2d})$, we have
\begin{equation}
\Lambda_{k,K_{\delta}}\big(\ba^{\perp}_{n,1}(\ba(t,\cdot))\big) < C_{T,K_{\delta}} \text{ and } \|\ba^{\perp}_{n,1}(\ba(t,\cdot))\|_{\rL^\infty} < C_{T,\delta}
\end{equation}
\end{proposition}
\begin{proof}
For both $\ba^\perp_\tp$ and $\ba^\perp_\ts$, the pseudo-operators $\cL^{-1}_{\tp,0}$, and $\cL^{-1}_{\ts,0}$ are bounded on $K_{\delta}$ as $|P| > 0$ as $|\cL^{-1}_{n,0}|_{\rL^\infty} \leq C \delta^{-1}$. Then by proposition~\ref{prop:L_bd} we the result.
\end{proof}
\begin{proposition}\label{prop:ew:bound}
Consider the Elastic wave equation with a forcing term.
\begin{equation}\label{eq:ew:force}
\begin{cases}
\rho(x)\partial_{t}^2 \bu - (\lambda + 2\mu)\nabla(\Div\bu) + \mu\Curl\Curl\bu  = \bF(t,\bx)& \\
u(0,\bx) = \bu_0^\epsilon, &\bx \in \mR^d \\
u(t,\bx) = \bu_1^\epsilon, &\bx \in \mR^d
\end{cases}.
\end{equation}
Let $T > 0$, and let $\bu_0(t,\bx) \in C^\infty([0,T])\times H^1_0(\mR^{d})$. For each $t\in[0,T]$ Then we have the following estimate:
\begin{multline}~\label{eq:est}
\|\rho\partial_t \bu\|_{\rL^2} + \|(\lambda + 2\mu)\Div\bu\|_{\rL^2} + \|\mu\Curl\bu\|_{\rL^2}\leq \\
C_T\left(\frac{1}{\epsilon}\|\bu(0,\cdot)\|_{\rE} + \int_0^t\|\bF(s,\cdot)\|_{\rL^2}\ud\ds\right) 
\end{multline}
In particular,
\begin{equation}
\sup_{t\in[0,T]}\|\bu(t,\cdot)\|_\rE \leq C_T \left( \|\bu_0^\epsilon\|_\rE + \epsilon\int_0^T\|\bF(s,\cdot)\|_{\rL^2}\ud\ds\right)
\end{equation}
\end{proposition}
\begin{proof}
This is a standard estimate, dotting Eq.~\eqref{eq:est} with $\partial_t \bu$ and integrating over space we have
\begin{equation}
\frac{\partial_t}{2} \int_{\mR^d}\rho(x) |\partial_t \bu|^2 + (\lambda + 2\mu)|\Div\bu|^2 + \mu|\Curl\bu|^2\ud\bx \leq \int_{\mR^d} |\bu_t\cdot\bF| \ud\bx
\end{equation}
The right-hand side can then be estimated by
\begin{equation}
\int_{\mR^d} |\bu_t\cdot\bF| \ud\bx \leq \frac{1}{2}\int_{\mR^d}|\partial_t\bu|^2 + |\bF|^2\ud\bx
\end{equation}
adding the missing terms to apply Gronwall's give the bound
\begin{equation}
e^t(\|\rho(x)\bu_1\|_{\rL^2}^2 + \|(\lambda+\mu)\nabla\cdot\bu_0\|_{\rL^2}^2 + \|\mu\Curl\bu\|_{\rL^2}^2) + \int_0^t e^{t-s} \int_{\mR^d}|\bF(s,\bx)|^2\ud \bx \ud s
\end{equation}
Taking max over $\rho, \lambda, \mu$ and over T we arrive at the estimate
\end{proof}
\begin{proposition}\label{prop:ew:fga:bd}
We have
\begin{equation}
\|(\partial_{t}^2 - \cL)\tilde{\bu}_{\FGA,1} \|_\rE \leq \epsilon C_{T,\delta}
\end{equation}
\end{proposition}
\begin{proof}
Plugging $\tilde{\bu}_{\FGA,1}$ into~\eqref{cp:eq:1} gives, 
\begin{multline}
(\partial_{t}^2 - \cL) \tilde{\bu}_{\FGA,1}(t,\bx) \\
= (2\pi\epsilon)^{-3d/2}\int_{\mR^{3d}} \chi_\delta\sum_{n,m} \epsilon^{m-2}\cL_{n,m}\Big(\ba_{n,0} + \epsilon (\ba_{n,1} + \ba_{n,1}^\perp)\Big) G_n^\epsilon \ud \by\ud \bp \ud\bq
\end{multline}
Expanding and simplifying yields;
\begin{multline}
(\partial_{t}^2 - \cL)\tilde{\bu}_{\FGA,1}(t,\bx) \\
=  (2\pi\epsilon)^{-3/2}\sum_{n}\chi_\delta G_n^\epsilon\ud \by\ud\bq\ud\bp\int_{\mR^{3d}}\epsilon^{-2}\cL_{n,0}(\ba_{n,0}) \\
+ \epsilon^{-1}\cL_{n,0}(\ba_{n,1}) +  \epsilon^{-1}\cL_{n,1}(\ba_{n,0}) + \cL_{n,1}(\ba_{n,1}) + \cL_{n,2}(\ba_{n,0} + \epsilon\ba_{n,1}) \\
-(\epsilon^{-1}\cL_{n,0} + \cL_{n,1} + \epsilon\cL_{n,2})\cL_{n,0}^{-1}((\Id - \Pi_n)(\cL_{n,1}(\ba_{n,0}\hat{\bd{N}_n})))\\
\end{multline}
Direct cancellation yields
\begin{multline}
(\partial_{t}^2 - \cL) \tilde{\bu}_{\FGA,1}(t,\bx) \\
=  (2\pi\epsilon)^{-3/2}\sum_{n}G_n^\epsilon\ud \by\ud\bq\ud\bp\int_{\mR^{3d}} R_0(t,\bp,\bq) + \epsilon R_1(t,\bp,\bq)
\end{multline}
where
\begin{align}
R_0(t,\bp,\bq) &= \cL_{n,1}(\ba_{n,1}) + \cL_{n,2}(\ba_{n,0}) -\cL_{n,1}\ba_{n,1}^\perp \\
R_1(t,\bp,\bq) &= \cL_{n,2}(\ba_{n,1}) - \cL_{n,2}\ba_{n,1}^\perp
\end{align}
Then by Propositions ~\ref{eq:fio_bd} and~\ref{prop:a_perp}
\begin{equation}
\|(\partial_{t}^2 - \cL) \tilde{\bu}_{\FGA,1}(t,\cdot)\|_{\rL^2} \leq C_{T,\delta} (\|R_0(t,\cdot)\|_{\rL^{\infty}} + \epsilon \|R_1(t,\cdot)\|_{\rL^{\infty}})
\end{equation}
By prop.~\ref{prop:a_perp},~\ref{prop:L_bd} $R_0$ and $R_1$ are bounded.
\end{proof}

\begin{proposition}\label{prop:init:est}
Let $\bu$ solve the Cauchy problem \eqref{cp:eq:1}.  If $\bu_{\FGA,0}$ is the first order FGA approximation~\eqref{eq:fga_sol}, then we have the following estimate with the initial conditions.
\begin{equation}
\|\bu(0,\bx) - \bu_{\FGA}(0,\bx)\|_{\rE} \leq \epsilon C_T
\end{equation}
\end{proposition}
\begin{proof}
First computing the following;
\begin{align}
\partial_t\ba_{n}(0,\by,\bq,\bp) &= \alpha(\by,\bq,\bp)\left(\ds\frac{\partial c(\bq)\cdot \bp}{|\bp|} - d\right)\hat{\bd{n}}_n
\end{align}
For estimating $\bu(0,\bx) - \bu_{\FGA,0}(0,\bx)$ in the energy norm, we can write $\partial_t u(0,x) = u_1^\epsilon(\bx)$. In terms of the FIO this gives;
\begin{multline}\label{eq:dt:FIO}
\int_{\mR^{3d}}u_0^\epsilon(\by)\ds\frac{\I}{\epsilon}\Phi_n(0,\bx,\by,\bq,\bp)G_n^\epsilon(0,\bx,\by,\bq,\bp) \ud\by\ud\bq\ud\bp  \\
= \int_{\mR^{3d}}\bu_1^\epsilon(\by)G_n^\epsilon(0,\bx,\by,\bq,\bp)\ud\by\ud\bq\ud\bp.
\end{multline}
Then
\begin{multline}
|\bu_1^\epsilon(0,\bx) - \partial_t \bu_{\FGA,0}(0,\bx)| = (2\pi\epsilon)^{-d/2}\Big|\sum_n\int_{\mR^{3d}} \big(\ds\frac{1}{2}(\bu_1^\epsilon(y)\cdot\hat{\bd{n}})\hat{\bd{n}} \\
- \partial_t\ba_{n}(0,\by,\bq,\bp) - \ba_{n}\ds\frac{\I}{\epsilon}\Phi_n(0,\bx,\by,\bq,\bp) \big)G_n^\epsilon(0,\bx,\by,\bq,\bp)\ud\by\ud\bq\ud\bp
\end{multline}
For one terms this is after plugging in~\ref{eq:dt:FIO},
\begin{multline}
\int_{\mR^{3d}} \Big(u_0^\epsilon(\by)\ds\frac{\I}{2\epsilon}\Phi_n(0,\bx,\by,\bq,\bp) - \alpha(\by,\bq,\bp)\left(\ds\frac{\partial c(\bq)\cdot \bp}{|\bp|} - d\right)\hat{\bd{n}}_n \\
 - 2^{d/2}\alpha(\by,\bq,\bp)\hat{\bd{n}}_n\ds\frac{\I}{\epsilon} \Phi_n(0,\bx,\by,\bq,\bp)\Big) \\
\times G_n^\epsilon(0,\bx,\by,\bq,\bp)\ud \bx\ud\by\ud\bq\ud\bp
\end{multline}
Plugging in $\alpha_n$ from~\eqref{eq:alpha} and summing over the wavefields and branches gives
\begin{multline}
\bu_1^\epsilon(\bx) - \partial_t \bu_{\FGA,0}(0,\bx) = \\
-\sum_n\int_{\mR^{3d}}\frac{1}{2c_{n}\abs{\bp}^3}\Bigl(\bd{u}_0^\epsilon(\bd{y})c_{n}\abs{\bd{p}}\pm
  \I\epsilon\bd{u}_1^\epsilon(\bd{y}) \Bigr)\cdot\hat{\bd{n}}_{n} \\
	\times\left(\ds\frac{\partial c(\bq)\cdot \bp}{|\bp|} - d\right)\hat{\bd{n}}_nG_n^\epsilon(0,\bx,\by,\bq,\bp) \ud\by\ud\bq\ud\bp
\end{multline}
By Proposition~\eqref{eq:fio_bd} we can arrive at the estimate
\begin{align}
\| \bu_1^\epsilon - \partial_t \bu_{\FGA,0}(0,\cdot)\|_{\rL^2} \leq C_T
\end{align}
For the Div term,
\begin{multline}
\Div\bu_\FGA(0,\bx)\ud\bx \\ 
= (2\pi\epsilon)^{-d/2}\sum_{n} \int_{\mR^{3d}}\Div\big(\ba_{n}(0,\by,\bq,\bp)G_n^\epsilon(0,\bx,\by,\bq,\bp)\big)\ud\by\ud\bq\ud\bp \\
= (2\pi\epsilon)^{-d/2}\sum_{n}\int_{\mR^{3d}}\frac{\I}{\epsilon}\ba_{n}\cdot\left(\bP_n + (\bx-\bQ_n)\right)G_n^\epsilon(0,\bx,\by,\bq,\bp) \ud\by\ud\bq\ud\bp
\end{multline}
applying the operators and integration by parts gives
\begin{equation}
\int_{\mR^d}\left(\ds\frac{\I}{\epsilon}\ba_{n}\cdot\bp_n -\partial_{\bz}(Z^{-1}\ba )\right)G_n^\epsilon(0,\bx,\by,\bq,\bp)\bx
\end{equation}
Now consider the difference, $\Div\bu_0^\epsilon - \Div\bu_{\FGA,0}(t,\bx)$.  Writing in terms of the FIO and writing one term
\begin{multline}
\int_{\mR^{3d}}\left(\ds\frac{\I}{\epsilon}\big(\bu_0^\epsilon(\by) - \ba(0,\by,\bq,\bp)\big)\cdot\bP_n +\partial_{\bz}(Z^{-1}(\bu_0^\epsilon(\by) - \ba(0,\by,\bq,\bp)) )\right)\\
\times G_n^\epsilon(0,\bx,\by,\bq,\bp)\ud\by\ud\bq\ud\bp
\end{multline}
With $\ba(0,\by,\bq,\bp) = \alpha_n(\by,\bq,\bp)\hat{\bd{n}}$ and summing over $n$ we have;
\begin{multline}
\Div\bu_0^\epsilon(\bx) - \Div\bu_{\FGA,0}(0,\bx) = \\
\sum_n\int_{\mR^d}\partial_{\bz}(Z^{-1}(\bu_0^\epsilon(\by) - \alpha_n(\by,\bq,\bp)\hat{\bd{n}})G_n^\epsilon(0,\bx,\by,\bq,\bp)\bx
\end{multline}
Again, by Proposition~\eqref{eq:fio_bd} we arrive at the estimate
\begin{align}
\|\Div\bu_0^\epsilon - \Div\bu_{\FGA,0}(t,\cdot)\ud\bx\|_{\rL^2} \leq C_T
\end{align}
The Curl term has a similar estimate as the Div term.  These 3 estimates show the result.
\end{proof}

\begin{theorem}\label{thm:1}
Let $\{u_0^\epsilon\}$ be a family of asymptotically high frequency initial conditions, and let $\bu$ solve the Cauchy problem \eqref{cp:eq:1}.  If $\bu_{\FGA,0}$ is the first order FGA approximation~\eqref{eq:fga_sol}, then for a given $T$ and any $t\in [0,T]$, $\delta > 0$ and sufficiently small $\epsilon$, we have
\begin{equation}
\sup_{t\in[0,T]}\|\bu(t,\cdot) - \bu_{\FGA,0}(t,\cdot)\|_{\rE} \leq \epsilon C_{T,\delta}
\end{equation}
\end{theorem}
\begin{proof}
\begin{equation}
\|\bu - \bu_{\FGA,0}\|_{\rE} \leq +  \|\bu - \tilde{\bu}_{\FGA,1}\|_{\rE} + \|\tilde{\bu}_{\FGA,1} - \bu_{\FGA,1}\|_{\rE} + \|\bu_{\FGA,1} - \bu_{\FGA,0}\|_{\rE} 
\end{equation}
For the first term define the quantity $e = \bu - \tilde{\bu}_{\FGA,1}$. By prop.~\ref{prop:ew:bound}
\begin{equation}
\|e\|_{\rE} \leq C_{T,\delta}\big(\|e(0,\cdot)\|_{\rE} + \epsilon\int_0^t\|R_0\|_{\rL^2}\ud s\big) + \cO(\epsilon^2)
\end{equation}
Prop~\ref{prop:init:est} shows that $\|e(0,\cdot)\|_{\rE} \leq \epsilon C_{T,\delta}$. Also, prop. \ref{prop:a_perp} and~\ref{prop:L_bd} show $\|R_0\|_{\rL^2} \leq C_{T,\delta}$.  Thus the first is estimated at the correct order. For the second term, is $\cO(\epsilon^\infty)$ by~\ref{prop:e_inf}, and the last term is estimated to the desired order by prop.~\ref{eq:a_bd}.
\end{proof}
As a consequence of the energy norm, we have the following corollary;
\begin{corollary}\label{cor:1} Under the same assumptions as Theorem~\ref{thm:1}, define $U_\FGA$ as
\begin{equation}
U_\FGA = (\partial_t\bu_\FGA,\Div\bu_\FGA, \Curl\bu_\FGA)
\end{equation}
Define $X_{\FGA}$ as in eq.~\ref{eq:xfga} and let $X$ be the solution to the hyperbolic system~\ref{hyper_decomp}, then we have the estimate
\begin{equation}
\sup_{t\in[0,T]}\|X_{\FGA,0}(t,\cdot) - U_{\FGA,0}(t,\cdot)\|_{\rL^2} \leq \epsilon C_{T,\delta}
\end{equation}
\end{corollary}
\begin{proof}
Applying the triangle inequality gives the immediate result.
\begin{equation}
\|X_{\FGA,0} - U_{\FGA,0}\|_{\rL^2} \leq \|X - X_{\FGA,0}\|_{\rL^2} +  \|X - U_{\FGA,0}\|_{\rL^2} \leq \epsilon C_{T,\delta} + \|\bu - \bu_{\FGA,0}\|_{\rE}
\end{equation}
Where $\|X - X_{\FGA,0}\|_{\rL^2}$ is bounded from~\eqref{eq:fga:efd}.
\end{proof}

\section*{Acknowledgments} 
Both authors were partially supported by the NSF grant DMS-1818592.

\bibliographystyle{abbrv}
\bibliography{paper,SEG,fga_ref}
\medskip
\medskip

\end{document}